\documentclass[a4paper,UKenglish,cleveref,autoref,thm-restate]{lipics-v2021}
\pdfoutput=1
\nolinenumbers

\author{Fugen Hagihara}{Graduate School of Science, Kyoto University, Japan}{hagihara@kurims.kyoto-u.ac.jp}{}{}
\author{Akitoshi Kawamura}{Research Institute for Mathematical Sciences, Kyoto University, Japan}{kawamura@kurims.kyoto-u.ac.jp}{}{}

\funding{This work was supported by JSPS KAKENHI Grant Numbers JP18H03203, JP23K28036, JP25KJ1559, ISHIZUE 2025 of Kyoto University, and JST SPRING Grant Number JPMJSP2110.}

\acknowledgements{We thank Kohki Baku at Faculty of Science, Kyoto University, for helping us find the explicit formula \eqref{eq:Baku_Koki} in \cref{ex:}. We also thank the anonymous referees for knowledgeable comments, which helped us clarify the explanation on previous work and our results. }

\authorrunning{F.~Hagihara, A.~Kawamura} 

\Copyright{Fugen Hagihara and Akitoshi Kawamura}



\fontsize{11pt}{11pt}\selectfont

\hideLIPIcs

\usepackage{amsmath,amssymb,amsthm}
\usepackage{tikz}

\usepackage[all]{xy}
\usepackage{array}
\usepackage{enumerate}

\definecolor{Green}{rgb}{0,0.753,0.192}
\definecolor{Pink}{rgb}{0.953,0,0.71}

\newcommand{\R}{\mathbb{R}}
\newcommand{\N}{\mathbb{N}}
\newcommand{\Q}{\mathbb{Q}}
\newcommand{\Z}{\mathbb{Z}}

\DeclareMathOperator{\sgn}{sgn}

\newcommand\Kettenbruch{%
  \operatornamewithlimits{%
    \mathchoice
     {\vcenter{\hbox{\huge $\mathrm{K}$}}}
     {\vcenter{\hbox{\Large $\mathrm{K}$}}}
     {\mathrm{K}}
     {\mathrm{K}}}}

\title{The Ultimate Signs of Second-Order Holonomic Sequences
\footnote{This is a full version of \cite{HK25} with detailed proofs.}
}

\begin{document}
\maketitle
\begin{abstract}
A real-valued sequence $f = \{f(n)\}_{n \in \N}$ is said to be second-order holonomic if it satisfies a linear recurrence $f (n + 2) = P (n) f (n + 1) + Q (n) f (n)$ for all sufficiently large $n$, where $P, Q \in \R(x)$ are rational functions. We study the ultimate sign of such a sequence, i.e., the repeated pattern that the signs of $f (n)$ follow for sufficiently large $n$.  For each $P$, $Q$ we determine all the ultimate signs that $f$ can have, and show how they partition the space of initial values of $f$.  This completes the prior work by Neumann, Ouaknine and Worrell, who have settled some restricted cases. As a corollary, it follows that when $P$, $Q$ have rational coefficients, $f$ either has an ultimate sign of length $1$, $2$, $3$, $4$, $6$, $8$ or $12$, or never falls into a repeated sign pattern. We also give a partial algorithm that finds the ultimate sign of $f$ (or tells that there is none) in almost all cases.
\end{abstract}

\section{Introduction}

Let $\N = \{ 0, 1, 2, \dots \}$ be the set of all natural numbers. 
A sequence $f = \{ f(n) \}_{n \in \N } \in \R^{\N}$ of real numbers is called a 
\emph{holonomic sequence} (of order $r \in \N$)
if there are real-coefficient rational functions $P_0, \dots , P_{r-1} \in \R(x)$ 
such that $f$ satisfies the linear recurrence 
\begin{equation}\label{eq:holonomic}
f(n + r) =  P_{r-1}(n) f(n+r-1) + \dots + P_0(n) f(n)  
\end{equation}
for all sufficiently large $n \in \N$. 
Holonomic sequences arise in various areas of mathematics.
For instance, solutions of linear differential equations with polynomial coefficients are generating functions of holonomic sequences \cite{Sta80} (see also \cite[Appendix B.4]{FS13}). Moreover, for ``proper hypergeometric terms'' $F(n, k)$ -- which typically involving binomial coefficients like $\binom{n}{k}$ -- the sum $f(n) = \sum_{k \in \Z} F(n, k)$ is holonomic, provided it converges for all $n \in \N$ \cite{PWZ96}. An algorithm that finds a holonomic recurrence satisfied by such sum of a given proper hypergeometric term is known as \emph{creative telescoping} \cite[Chapter~6]{PWZ96}.

An important computational problem concerning holonomic sequences 
is the \emph{Ultimate Sign Problem} \cite{NOW21}: 
Given (rational-coefficient) rational functions $P_0, \dots, P_{r-1} \in \Q(x)$ without poles in $\N$ 
and (rational-valued) initial values $f(0)$, \dots, $f(r-1) \in \Q$, 
find an ultimate sign, defined as follows, of the unique sequence $f$ 
having these initial values and satisfying \eqref{eq:holonomic} for all $n \in \N$, 
and an index $N \in \N$ at which this ultimate sign is reached. 
Although we assume that $f$ satisfies the recurrence \eqref{eq:holonomic} not only for $n \geq I$ for some $I \in \N$ but also for all $n$, it is not different in computability from the problem of finding the ultimate sign and the index $N$ from the coefficients $P_0, \dots, P_{r-1}$, initial values $f(I)$, \dots, $f(I+r-1)$ and $I$. 

\begin{definition}
A sequence $f \in \R^{\N}$ is said to have an
\emph{ultimate sign} $(s_0, \dots, s_{\tau-1}) \in \{ +, -, 0 \} ^*$ 
at $N \in \N$
if $\sgn f(n) = s_{n \bmod \tau}$ for all $n \geq N$, 
where $\sgn \colon \R \to \{ +, -, 0 \}$ is the function that maps each real number to its sign. 
\end{definition}

For instance, the sequence $\{ (-1)^n (n-2) \}_{n \in \N} = -2, 1, 0, -1, 2, -3, \dots$ has 
the ultimate sign $(+, -)$ at $3$.
Note that if $f$ has the ultimate sign $s$ at $N$, 
then it also has any repetition of $s$ as an ultimate sign, and it does so at any index $\geq N$; 
but we could of course ask for the \emph{shortest} ultimate sign $s$ and the \emph{least} index $N$ 
without changing the computability of the problem.

The Ultimate Sign Problem is a generalization of several important problems about signs of holonomic sequences.
One of the most famous problems is the \emph{Skolem Problem}, 
which asks whether $f(n) = 0$ for some $n$ 
(see \cite[\S~4]{OW12} for an argument that it reduces to the Ultimate Sign Problem). 
Its decidability has been studied for almost 90 years \cite{HHHK05}. 
The \emph{Positivity Problem} asking whether $f(n) > 0$ for all $n$ and 
the \emph{Ultimate Positivity Problem} asking whether $f$ has the ultimate sign $(+)$ 
are also well studied, with applications to automated inequality proving \cite{GK05}; 
see also subsequent works \cite{KP10,Pil13,PS15}
and a SageMath implementation \cite{Nus22}. 

When the coefficients $P_0, \dots, P_{r-1}$ are constant, $f$ is called a C-finite sequence (or a linear recurrence sequence). 
The Skolem Problem for C-finite sequences of order $r \leq 4$ \cite{Tij84,Ver85} and 
the (Ultimate) Positivity Problem for C-finite sequences of order $r \leq 5$ \cite{OW14} 
are known to be decidable, whereas the decidability for higher order C-finite sequences is open. 

For holonomic sequences, when $r = 1$ (i.e., when $f$ is a hypergeometric sequence), the Ultimate Sign Problem is easy 
since for given $ P_0 \in \Q(x)$, 
we can effectively compute
an index $N \in \N$ such that $P_0(n)$ has a constant sign for $n \geq N$. 
When $r=2$, i.e., when $f$ satisfies a recurrence of the form
\begin{equation}\label{eq:(PQ)holonomic}
f(n+2) = P(n) f(n+1) + Q(n) f(n), 
\end{equation}
the decidability of Skolem and (Ultimate) Positivity Problem for 
some subclasses is known in the context of the Membership Problem \cite{NPSW22} and the Threshold Problem \cite{Ken24}, respectively. 
\cite[Theorem~7]{NOW21} shows that the Ultimate Sign Problem for another subclass is computable. However, the computability for general second-order holonomic sequences remains unknown. 
To make progress on this open problem, we study the ultimate signs of all second-order holonomic sequences. 

Our first main contribution is 
to classify all pairs $(P, Q) \in \R(x) ^2$ by the ultimate signs $f$ can have (\cref{cor:Spq}), 
and show how the ultimate signs partition the space of initial values of $f$ (\cref{thm:main}). This result resolves all remaining cases in \cite[Theorem~1]{NOW21}, which handles the 
restricted case where $P$, $Q$ are polynomials, $P$ is non-constant and $\deg Q \leq \deg P$. 
In addition, this result implies that when $P$, $Q$ have rational coefficients, 
the shortest ultimate sign of $f$, if it has one, is either of length $1$, $2$, $3$, $4$, $6$, $8$ or $12$ (\cref{cor:}). 

Our second contribution is to give a partial algorithm that
solves the Ultimate Sign Problem for second-order holonomic sequences 
and halts on almost all inputs (\cref{thm:procedure}). 
This extends a similar result \cite[Theorem~3]{NOW21} for the restricted case mentioned above. 
This result can be also stated as a reduction theorem: for second-order holonomic sequences, the Ultimate Sign Problem Turing-reduces to the Minimality Problem, which asks the minimality of a given $f$, i.e., whether $f(n) / g(n) \to 0$ for all linearly independent solutions $g$ of the same recurrence. 
In this sense our result extends \cite[Theorem~3.1]{KKL+21}, which shows that the Positivity Problem Turing-reduces to the Minimality Problem. 
Note that, unfortunately, the decidability of Minimality Problem is unknown whereas many researchers numerically calculate minimal holonomic sequences and apply them to numerical analysis of some special functions (for example \cite{Gau67,DST10}). 
In addition, combining this partial algorithm with creative telescoping allows us to determine all $n \in \N$ for which an inequality of the form $\sum_{k} F(n,k) > \sum_k G(n,k)$ holds, for some proper hypergeometric terms $F(n,k)$ and $G(n, k)$ (\cref{ex:ct + our pa}).

As a byproduct of our arguments, we amend some gaps in the proof of \cite{NOW21}, slightly modifying its Theorem~7 (\cref{thm:NOW21}). This will be discussed in Section~\ref{sec:total_algorithm}. 

\paragraph*{Related work}

A lot of previous works describe their results in terms of continued fractions, which have a strong connection to second-order holonomic sequences. We illustrate the connection between those works and one of our main theorems in \cref{sec:cfrac,sec:monotonic_convergence_thm}. 

Not only the ultimate signs, but also other periodicities of signs of holonomic (or C-finite) sequences are investigated. Closely related to the Skolem Problem, the periodicity of the zeros of C-finite (and for some holonomic) sequences is well-known as the Skolem-Mahler-Lech theorem \cite{BBY12}. 
Almagor et al.\ \cite{AKK+21} give some sufficient conditions for C-finite sequences to have an ``almost periodic sign'', a loose property of sign periodicity.

Kooman \cite{Koo07} studies the asymptotic behaviour of complex solutions of the recurrence~\eqref{eq:(PQ)holonomic}, 
where $P$ and $Q$ are not necessarily rational functions. 
His results helped us see the big picture of our main theorems.

\section{Results}\label{sec:result}

The Ultimate Sign Problem asks about the ultimate signs of $f$ that satisfies \eqref{eq:(PQ)holonomic} for all $n$. Such $f$ is identified by the coefficient pair $(P, Q)$ and the initial value $(f(0), f(1))$. 

\begin{definition}
Let $P$, $Q \in \R(x)$ be rational functions without poles in $\N$. 
A sequence $f \in \R^{\N}$ is $(P, Q)$-\emph{holonomic} 
if it satisfies \eqref{eq:(PQ)holonomic}. 
The pair $(f(0), f(1)) \in \R^2$ is called the \emph{initial value} of $f$. 
\end{definition}

The Ultimate Sign Problem for $(0, Q)$- or $(P, 0)$-holonomic sequences is easy, 
so we assume $P \neq 0$ and $Q \neq 0$. 
By shifting the index by finitely many terms, 
we may assume that $P$, $Q$ have no zeros in $\N$. 
This shifting changes the ultimate sign and the initial value of $f$ in such a simple way that
it does not affect the computability of the Ultimate Sign Problem. 
We adopt this assumption in all the following theorems.

\subsection{Ultimate signs}
\label{subsection: main result classification}

In this section, for each type of $(P, Q) \in (\R(x)\setminus\{0\})^2$ in the following \cref{def:classification}, we list all the ultimate signs that $(P, Q)$-holonomic sequences $f$ can have. Moreover, we show how the ultimate signs partition the space of initial values of $f$. 
For $R \in \R(x) \setminus \{ 0 \}$, let $\deg R$ denote $d \in \Z$ satisfying $|R(x)| = \Theta(x^d)$. 
By the \emph{ultimate sign} of $R$ we mean that of $\{ R(n) \}_{n \in \N}$.

\begin{definition}\label{def:classification}
We classify $(P, Q) \in (\R(x) \setminus \{ 0 \})^2$ into the following types. Let $d := \deg \frac{Q(x)}{P(x)P(x-1)}$ and $(s) \ (s \in \{ +, -\})$ be the ultimate sign of $\frac{Q(x)}{P(x)P(x-1)}$. 

\begin{itemize}
\item If $s = +$ and $d > 2$, then we say that $(P, Q)$ is of \emph{$\infty$-$O$ loxodromic type}.

\item If $s = +$ and $d \leq 2$, then we say that $(P, Q)$ is of \emph{$\infty$-$\Omega$ loxodromic type}.

\item If $s = -$ and $d \leq 0$, then let $\alpha_0, \alpha_1, \alpha_2$ be real numbers satisfying  
\begin{equation}\label{eq:Taylor_ex}
\frac{Q(x)}{P(x)P(x-1)} = \alpha_0 + \frac{\alpha_1}x + \frac{\alpha_2}{x^2} + O(x^{-3}). 
\end{equation}
\begin{itemize}

\item If $(\alpha_0, \alpha_1, \alpha_2) \geq (-\frac{1}{4}, 0, -\frac1{16})$ in lexicographic order, then we say that $(P, Q)$ is of \emph{hyperbolic type}. 

\item Otherwise, $\alpha_0 \leq  -\frac{1}{4}$, so there is a real number $\theta \in [0, \frac12)$ such that $\alpha_0 = -\frac1{4 \cos^2 \theta \pi}$.

\begin{enumerate}[(1)]
\item If $\theta$ is a positive rational number and $\alpha_1 = 0$, then we say that $(P, Q)$ is of \emph{$\theta$-$O$ elliptic type}.

\item \label{item:otherwise} Otherwise, we treat $(P, Q)$ together with the next case. 
\end{enumerate}	
\end{itemize}	

\item If $s = -$ and $d = 1, 2$, or it is the case of \eqref{item:otherwise} above, then we say that $(P, Q)$ is of \emph{$\Q$-$\Omega$ elliptic type}. 

\item If $s = -$ and $d > 2$, then we say that $(P, Q)$ is of \emph{$\frac12$-$O$ elliptic type}.
\end{itemize}
\end{definition}

Considering the rational function $\frac{Q(x)}{P(x)P(x-1)}$ in the above definition is reasonable, since it naturally appears when we normalize $P$ to 1. For $(P, Q)$-holonomic sequences $f$, the sequence $\left\{ \frac{f(n)}{P(n-2) \dotsm P(-1)} \right\}_{n \in \N} $ is $\left( 1, \frac{Q(x)}{P(x)P(x-1)} \right)$-holonomic. 

This classification consists of the distinctions between \emph{loxodromic type} ($\infty$-$O$ loxodromic type and $\infty$-$\Omega$ loxodromic type), hyperbolic type and \emph{elliptic type} ($\theta$-$O$ elliptic type and $\Q$-$\Omega$ elliptic type), and between \emph{$O$ type} ($\infty$-$O$ loxodromic type and $\theta$-$O$ elliptic type) and \emph{$\Omega$ type} ($\infty$-$\Omega$ loxodromic type and $\Q$-$\Omega$ elliptic type). 
The highly non-trivial border between hyperbolic type and elliptic type is well-studied in the context of the convergence of continued fractions (\cref{thm:cfrac_cconverge}). 

The terminologies of ``$O$'' and ``$\Omega$'' come from big O and $\Omega$ notations. They represent whether $\frac{Q(x)}{P(x)P(x-1)}$ is near or apart from a certain value ($\infty$ for loxodromic type, $-\frac1{4\cos^2 \theta \pi}$ for $\theta$-$O$ elliptic type and $-\frac1{4\cos^2 q \pi}$  for all $q \in (0, \frac12] \cap \Q$ for $\Q$-$\Omega$ elliptic type).

The terminologies of loxodromic, hyperbolic and elliptic come from the classification of 
linear fractional transformations. 
If $P$ and $Q$ are constant, the linear fractional transformation $z \mapsto \frac1{P+Qz}$ maps the ratio $f (n) / f (n+1)$ between the two neighbouring terms of the $(P, Q)$-holonomic sequence
to the next ratio $f (n + 1) / f (n + 2)$, and is said to be elliptic, parabolic, hyperbolic and loxodromic 
when $\frac{Q}{P^2}$ is in
$(-\infty, -\frac 1 4)$, $\{-\frac 1 4\}$, $(-\frac 1 4, 0)$ and $(0, \infty)$, respectively 
(with slight variations among authors --
some (cf.\ \cite[\S~4.1.3]{LW08}) treat hyperbolic as a subclass of loxodromic, 
while some (cf.\ \cite[\S~4.7]{Rat19}) treat loxodromic as a subclass of hyperbolic).

Now we explicitly describe the set 
\[
S _{P, Q} := 
\left\{ 
s \in \{ +, -, 0 \}^* \cup \{ \bot \}
\middle|
\begin{array}{l}
\text{There exists a non-zero $(P, Q)$-holonomic sequence}
\\
\text{that has the shortest ultimate sign $s$.}
\end{array}
\right\},
\]
where 
a sequence without any ultimate sign is considered to have the shortest ultimate sign $\bot$. 
We exclude the zero sequence since it obviously has the ultimate sign $(0)$. 
This is a simple and direct corollary of \cref{thm:main}.
In the following corollary and theorem, we assume that $P$ has ultimate sign $(+)$. For $P$ with ultimate sign $(-)$, instead of a $(P, Q)$-holonomic sequence $f$, consider the $(-P, Q)$-holonomic sequence $\{ (-1)^n f(n) \}_{n \in \N}$.

\begin{corollary}\label{cor:Spq}
Let $P, Q \in \R(x)$ be rational functions without poles or zeros in $\N$, and suppose that $P$ has ultimate sign $(+)$. 
Define $S_{P, Q}  \subseteq \{+, -, 0 \}^* \cup \{ \bot \}$ as above.
\begin{itemize}
\item If $(P, Q)$ is of loxodromic type, $S_{P, Q} = \{ (+), (-), (+, -), (-, +)\}$.
\item If $(P, Q)$ is of hyperbolic type, $S_{P, Q} = \{ (+), (-)\}$.
\item If $(P, Q)$ is of $\frac{k}r$-$O$ elliptic type, where $r$ and $k$ are coprime positive integers, let 
\begin{align} \label{eq:sjtj}
s_{j} & = \biggl( \sgn \sin \frac{j - ik + 0.5}r \pi \biggr)_{i=0, \dots, (k \bmod 2)r+r-1} &&\in \{ +, - \}^*, 
\notag
\\
t_{j} & = \biggl( \sgn \sin \frac{j - ik}r \pi \biggr)_{i=0, \dots, (k \bmod 2)r+r-1}
&&\in \{ +, -, 0\}^*
\end{align}
for each $j = 0, \dots, 2r-1$, where $k \bmod 2$ is $0$ if $k$ is even and $1$ if $k$ is odd. 
\begin{itemize}
\item If $\frac{Q(x)}{P(x)P(x-1)}$ is constant, $S_{P, Q} = \{ s_j \mid j=0, \dots, 2r-1 \} \cup \{ t_j \mid j=0, \dots, 2r-1 \}$.
\item Otherwise, $S_{P, Q} = \{ s_j \mid j=0, \dots, 2r-1 \}$.
\end{itemize}
\item If $(P, Q)$ is of $\Q$-$\Omega$ elliptic type, $S_{P, Q} = \{ \bot \}$.
\end{itemize}
\end{corollary}

The value $0.5$ inside the sine function in $s_j$ can be replaced by any value between $0$ and $1$. 

For constant $P$, $Q$, this corollary is easily obtained since the recurrence \eqref{eq:(PQ)holonomic} can be solved explicitly.

The following main theorem shows how the space $\R^2 \setminus \{ 0 \}$ of the initial values of $f$ is partitioned into sets 
\[
I_{P, Q}(s) := \left\{ f_0 \in \R^2 \setminus \{ 0 \} \middle| 
\begin{array}{l}
\text{The $(P, Q)$-holonomic sequence with initial value $f_0$}
\\
\text{has ultimate sign $s$.}
\end{array}
\right\}
\] 
with $s \in \{ +,-,0\}^*$. 
Since the set $I_{P, Q} (s)$ is closed under linear combinations with positive coefficients, 
it is a convex linear cone and thus specified by 
an (open, closed or half-open) interval $p (I _{P, Q} (s))$ 
on the unit circle $S ^1$, where 
\begin{equation}\label{eq:projection}
p \colon \R^2 \setminus \{ (0, 0) \} \to S ^1 ; \ (x, y) \mapsto (x, y) / \sqrt{x^2 + y^2}
\end{equation}
is the projection. 
Thus, we will state the theorem by describing
how $S ^1$ is partitioned into intervals $p (I _{P, Q} (s))$. 
It is obvious that flipping the sign of the initial value flips each element of the ultimate sign, so that $I_{P, Q} (-s)$ is just $I_{P, Q}(s)$ flipped around the origin. 
We omit the parentheses and write $I_{P, Q}(+, -)$, say, for $I_{P, Q}((+, -))$. 

We suppose that $P$ has ultimate sign $(+)$ as in \cref{cor:Spq}.

\begin{theorem}\label{thm:main}
Let $P$, $Q \in \R(x)$ have no zeros or poles in $\N$, 
and suppose that $P$ has ultimate sign $(+)$. 
For each $s \in \{ +, -, 0 \}^*$, 
we write $p (I_{P, Q} (s))$ for the set of $f_0 \in S ^1$ such that 
the $(P, Q)$-holonomic sequence with initial value $f _0$ 
has the ultimate sign $s$.
\begin{enumerate}[(I)]
\item \label{item:loxodromic_nearly_infty}
If $(P, Q)$ is of $\infty$-$O$ loxodromic type, $S ^1$ is partitioned into
closed intervals \\
$p(I_{P, Q}(+, -))$, $p(I_{P, Q}(-, +))$ which have non-empty interiors and
non-empty open intervals $p(I_{P, Q}(+))$, $p(I_{P, Q}(-))$. 
\item \label{item:loxodromic_far_from_infty}
If $(P, Q)$ is of $\infty$-$\Omega$ loxodromic type, $S ^1$ is partitioned into
singletons $p(I_{P, Q}(+, -))$,
\\
$p(I_{P, Q}(-, +))$ and 
non-empty open intervals $p(I_{P, Q}(+))$, $p(I_{P, Q}(-))$. 
\item \label{item:hyperbolic}
If $(P, Q)$ is of hyperbolic type, $S ^1$ is partitioned into
half-open intervals $p(I_{P, Q}(+))$, $p(I_{P, Q}(-))$. 
\item \label{item:k/r-O_elliptic}
If $(P, Q)$ is of $\frac{k}r$-$O$ elliptic type, 
where $r$ and $k$ are coprime positive integers, 
define $s_j \in \{ +, - \}^*$ and $t_j \in \{ +, -, 0 \}^*$ as in \eqref{eq:sjtj} for each $j = 0$, \ldots, $2r-1$. 
\begin{itemize}
\item
If $\frac{Q(x)}{P(x)P(x-1)}$ is constant, 
$S ^1$ is partitioned into
$p(I_{P, Q}(t_0))$, $p(I_{P, Q}(s_0)), \dots, \break p(I_{P, Q}(t_{2r-1}))$, $p(I_{P, Q}(s_{2r-1}))$, 
arranged in this order (clockwise or anticlockwise), 
of which $p(I_{P, Q}(t_{j}))$ are singletons and $p(I_{P, Q}(s_{j}))$ are non-empty open intervals.
\item
Otherwise, 
$S ^1$ is partitioned into
non-empty half-open intervals
$p(I_{P, Q}(s_0)), \dots, \break p(I_{P, Q}(s_{2r-1}))$,
arranged in this order, 
where 
for each $j=0$, \ldots, $2r-1$, 
the intersection of the closures of
$p(I_{P, Q}(s_j))$ and $p(I_{P, Q}(s_{j+1}))$
(where $s_{2r} = s_0$)
belongs to 
$p(I_{P, Q}(s_{j+1}))$ if $\frac{Q(x)}{P(x)P(x-1)}$ is eventually increasing (i.e., increasing for sufficiently large $x$), 
and to $p(I_{P, Q}(s_j))$ if it is eventually decreasing. 
\end{itemize}
\item \label{item:Q-Omega_elliptic}
If $(P, Q)$ is of $\Q$-$\Omega$ elliptic type, 
then no non-zero $(P, Q)$-holonomic sequence has an ultimate sign. 
\end{enumerate}
\end{theorem}

If $(P, Q)$ is of $\frac12$-$O$ elliptic type, then $\frac{Q(x)}{P(x)P(x-1)}$ necessarily decreases eventually. 

In Parts \eqref{item:loxodromic_nearly_infty}, \eqref{item:loxodromic_far_from_infty}, \eqref{item:hyperbolic} and \eqref{item:k/r-O_elliptic}, the union of the boundaries of the sets $I(s)$ is a finite union of lines. Following \cite{NOW21}, which handles restricted cases of \eqref{item:loxodromic_far_from_infty} and \eqref{item:hyperbolic} with $\deg \frac{Q(x)}{P(x)P(x-1)} \leq -1$, we call these lines the \emph{critical lines}. 

\begin{example} \label{ex:unknown_slope}
It is known that the sequence $L_t = \{ L_t(n) \}_{n \in \N}$ of the values of Legendre polynomials
\begin{equation}
L_t(n) := 2^{-n} \sum_{k=0}^{\lfloor n/2 \rfloor} (-1)^k \binom{2n-2k}{n-k} \binom{n-k}{k} t^{n-2k} \in \Q[t]
\end{equation}
is a $(P_t(x), Q(x)) = \left( \frac{2x+3}{x+2} t , - \frac{x+1}{x+2} \right)$-holonomic sequence with initial value $(1, t)$, which can be verified by creative telescoping.
Let us apply \cref{thm:main} to this $(P_t(x), Q(x))$ for $t > 0$. 
Since $\frac{Q(x)}{P_t(x)P_t(x-1)} = -\frac1{4t^2} - \frac1{4t^2(2x+3)(2x+1)}$,
the pair $(P_t, Q)$ is of hyperbolic type if $t \geq 1$, of $\theta$-$O$ elliptic type if $t = \cos \theta \pi$ for some $\theta \in \Q \cap (0, \frac12)$, and of $\Q$-$\Omega$ elliptic type otherwise.
Therefore, all nonzero $(P_t, Q)$-holonomic sequences have an ultimate sign in the former two cases, while none of them does in the last case. 
Next, we illustrate the partition described in \cref{thm:main}~\eqref{item:k/r-O_elliptic} for $t = \cos \frac{\pi}3, \cos \frac{\pi}4$, i.e., $t= \frac12, \frac1{\sqrt{2}}$.
The shortest ultimate signs that nonzero $(P_t, Q)$-holonomic sequences can have are
\begin{align}\label{equation: ultimate sign example elliptic rational}
    s _0 = s_6 & = (+, -, -, -, +, +), &
    s _1 & = (+, +, -, -, -, +), &
    s _2 & = (+, +, +, -, -, -), 
 \notag
 \\ 
    s _3 & = (-, +, +, +, -, -), &
    s _4 & = (-, -, +, +, +, -), &
    s _5 & = (-, -, -, +, +, +)
\end{align} 
for $t = \frac12$, and
\begin{align}\label{equation: 8 ultimate signs}
    s _0  = s_8  &= (+, -, -, -, -, +, +, +), &
     s _1  &= (+, +, -, -, -, -, +, +), 
    \notag
    \\
    s _2  &= (+, +, +, -, -, -, -, +), &
    s _3  &= (+, +, +, +, -, -, -, -), 
    \notag
    \\
    s _4  &= (-, +, +, +, +, -, -, -), &
    s _5 & = (-, -, +, +, +, +, -, -), 
    \notag
    \\ 
    s _6  &= (-, -, -, +, +, +, +, -), &
    s _7  &= (-, -, -, -, +, +, +, +)
\end{align} 
for $t= \frac1{\sqrt{2}}$. Note that $\frac{Q(x)}{P_{t}(x)P_{t}(x-1)}$ is increasing. Then the partitions of the set $\R^2 \setminus \{ 0 \}$ of the initial values $(f(0), f(1))$ into $I_{P, Q}(s_j)$ are illustrated as in \cref{figure: example t=1/2 and t=1/sqrt2}.
In these figures, the exact slopes of the critical lines $l_j$ are unknown, although we can numerically approximate them to arbitrary precision using the method in \cref{ex:partial algorithm}.

\begin{figure}
\begin{center}
\begin{tikzpicture}[scale=0.89]
\draw[->] (-3.5,0) -- (3.5, 0) node[below] {\ \ \ $f(0)$};
\draw[->] (0,-3.7)-- (0,3.4) node[right] {$f(1)$};
\coordinate (O) at (0, 0);
\coordinate (A) at (3.2, 3.2*0.68821);
\coordinate (B) at (-3.2/7.2874, 3.2);
\coordinate (C) at (-3.4, -3.4*0.02832);
\coordinate (D) at (-3.4, -3.4*0.68821);
\coordinate (E) at (3.4/7.2874, -3.4);
\coordinate (F) at (3.2, 3.2*0.02832);

\fill[Green,opacity=0.1] (O)--(F)--(A)--cycle;
\draw[orange,line width=1.2pt] (O)--(A)node[right] { {$l_2$}};
\draw[Green] (2.3, 0.5) node {\small \colorbox{Green!2!white}{\textcolor{Green}{$I_{P, Q}(s_1)$}}};

\fill[orange,opacity=0.1] (O)--(A)--(3.2,3.2)--(B)--cycle;
\draw[blue,line width=1.2pt] (O)--(B);
\draw[orange] (A) node[above left] {\!\!\!\!\small \colorbox{orange!2!white}{$I_{P, Q}(s_2)$}};

\fill[Green,opacity=0.1] (O)--(C)--(D)--cycle;
\draw[orange,line width=1.2pt] (O)--(D);
\draw[Green] (-2.24, -0.6) node {\small \colorbox{Green!2!white}{\textcolor{Green}{$I_{P, Q}(s_4)$}}};

\draw[Green,line width=1.2pt] (O)--(C);
\fill[blue,opacity=0.1] (O)--(B)--(-3.4,3.2)--(C)--cycle;
\draw[blue] (-2.8,2.8) node[right] {\small \colorbox{blue!2!white}{$I_{P, Q}(s_3)$}};

\fill[orange,opacity=0.1] (O)--(D)--(-3.4,-3.4)--(E)--cycle;
\draw[blue,line width=1.2pt] (O)--(E) node[below right] { {$l_0$}};
\draw[orange] (-1.6, -3) node {\small \colorbox{orange!2!white}{$I_{P, Q}(s_5)$}};

\draw[Green,line width=1.2pt] (O)--(F) node[above right] {\textcolor{Green}{$l_1$}};
\fill[blue,opacity=0.1] (O)--(E)--(3.2,-3.4)--(F)--cycle;
\draw[blue] (3,-3) node[left] {\small \colorbox{blue!2!white}{$I_{P, Q}(s_0)$}};

\draw (0.08, -0.42) node[left] {O};
\draw[line width=1.3] (O) circle[radius=0.04];
\end{tikzpicture}
\begin{tikzpicture}[scale=0.89]
\draw[->] (-3.5,0) -- (3.5, 0) node[below] {\ \ \ $f(0)$};
\draw[->] (0,-3.7)-- (0,3.4) node[right] {$f(1)$};
\coordinate (O) at (0, 0);
\coordinate (A) at (3.2, 3.2*0.4931395);
\coordinate (AB) at (3.2/1.05645,3.2);
\coordinate (B) at (-3.2/3.62695, 3.2);
\coordinate (C) at (-3.4, -3.4*0.054375);
\coordinate (D) at (-3.4, -3.4*0.4931395);
\coordinate (DE) at (-3.4/1.05645,-3.4);
\coordinate (E) at (3.4/3.62695, -3.4);
\coordinate (F) at (3.2, 3.2*0.054375);

\fill[orange,opacity=0.1] (O)--(A)--(3.2,3.2)--(AB)--cycle;
\draw[orange] (2.68,1.8) node {\!\!\!\!\small \colorbox{orange!2!white}{$I_{P, Q}(s_2)$}};
\draw[Pink,line width=1.2pt] (O)--(AB) node[above] { {$l_3$}};

\fill[Green,opacity=0.1] (O)--(F)--(A)--cycle;
\draw[orange,line width=1.2pt] (O)--(A)node[right] { {$l_2$}};
\draw[Green] (2.4, 0.51) node {\small \colorbox{Green!2!white}{$I_{P, Q}(s_1)$}};

\fill[Pink,opacity=0.1] (O)--(AB)--(B)--cycle;
\draw[blue,line width=1.2pt] (O)--(B);
\draw[Pink] (1.3,2.8) node {\!\!\!\!\small \colorbox{Pink!2!white}{$I_{P, Q}(s_3)$}};

\fill[Green,opacity=0.1] (O)--(C)--(D)--cycle;
\draw[orange,line width=1.2pt] (O)--(D);
\draw[Green] (-2.6, -0.6) node {\small \colorbox{Green!2!white}{$I_{P, Q}(s_5)$}};

\draw[Green,line width=1.2pt] (O)--(C);
\fill[blue,opacity=0.1] (O)--(B)--(-3.4,3.2)--(C)--cycle;
\draw[blue] (-2.8,2.8) node[right] {\small \colorbox{blue!2!white}{$I_{P, Q}(s_4)$}};

\fill[orange,opacity=0.1] (O)--(D)--(-3.4,-3.4)--(DE)--cycle;
\draw[orange] (-2.68, -2) node {\small \colorbox{orange!2!white}{$I_{P, Q}(s_6)$}};
\draw[Pink,line width=1.2pt] (O)--(DE);

\fill[Pink,opacity=0.1] (O)--(DE)--(E)--cycle;
\draw[blue,line width=1.2pt] (O)--(E) node[below right] { {$l_0$}};
\draw[Pink] (-1,-3) node {\!\!\!\!\small \colorbox{Pink!2!white}{$I_{P, Q}(s_7)$}};

\draw[Green,line width=1.2pt] (O)--(F) node[above right] {$l_1$};
\fill[blue,opacity=0.1] (O)--(E)--(3.2,-3.4)--(F)--cycle;
\draw[blue] (3,-3) node[left] {\small \colorbox{blue!2!white}{$I_{P, Q}(s_0)$}};

\draw (0.088, -0.5) node[left] {O};
\draw[line width=1.3] (O) circle[radius=0.04];
\end{tikzpicture}
\end{center}
\caption{left: The partition of the space of the initial values into $I_{P_{1/2}, Q}(s_j)$ for $s_j$ in \eqref{equation: ultimate sign example elliptic rational}.
\\
right: The partition of the space of the initial values into $I_{P_{1/ \sqrt{2}}, Q}(s_j)$ for $s_j$ in \eqref{equation: 8 ultimate signs}.
\\
In both pictures, the critical half line between $I_{P_t, Q}(s_j)$ and $I_{P_t, Q}(s_{j+1})$ belongs to $I_{P, Q}(s_{j+1})$ for each $j$, where $t= \frac12, \frac1{\sqrt{2}}$, respectively.
\label{figure: example t=1/2 and t=1/sqrt2}}
\end{figure}

\end{example}

\begin{example}\label{ex:}
Let $P(x) = \frac{x+2}{x+1}$ and $Q(x) = - \frac{x+3}{x+1}$, 
so that $\frac{Q(x)}{P(x)P(x-1)} = -1 + \frac2{x^2 + 3x + 2}$ is decreasing 
and $(P, Q)$ is $\frac13$-$O$ elliptic. 
By Theorem~\ref{thm:main}~\eqref{item:k/r-O_elliptic}, 
non-zero $(P, Q)$-holonomic sequences $f$ in this case
have ultimate signs $s_0, \dots, s_6$ in \eqref{equation: ultimate sign example elliptic rational},
and the partition of the space of the initial values into the sets $I_{P, Q}(s_j)$ is illustrated as \cref{figure: example elliptic rational}.
In this illustration, we know the exact values of the slopes of the critical lines since, for this $P$ and $Q$, we can solve the recurrence \eqref{eq:(PQ)holonomic} explicitly:
\begin{equation}\label{eq:Baku_Koki}
f(n) = \begin{cases}
(-1)^m \left( \left( \frac72 m + 1 \right) f(0) - m f(1) \right) & \text{if} \ n = 3m, \\
(-1)^m \left( m f(0) + (m+1) f(1) \right) & \text{if} \ n = 3m+1, \\
(-1)^{m+1} \left( \left( \frac52 m + 3 \right) f(0) - 2(m+1) f(1) \right) & \text{if} \ n = 3m+2. 
\end{cases}
\end{equation}
\begin{figure}
\begin{center}
\begin{tikzpicture}[scale=0.97]
\draw[->] (-3.5,0) -- (3.5, 0) node[below] {\ \ \ $f(0)$};
\draw[->] (0,-3.7)-- (0,3.4) node[left] {$f(1)$};
\coordinate (O) at (0, 0);
\coordinate (A) at (6.4/7, 3.2);
\coordinate (B) at (-3.2, 3.2);
\coordinate (C) at (-13.6/5, -3.4);
\coordinate (D) at (-6.8/7, -3.4);
\coordinate (E) at (3.4, -3.4);
\coordinate (F) at (12.8/5, 3.2);

\draw[orange] (A) node[above] { slope $\frac72$};
\fill[orange,opacity=0.1] (O)--(F)--(A)--cycle;
\draw[orange] (A) node[below right] {\!\!\!\!\small \colorbox{orange!2!white}{$I_{P, Q}(s_2)$}};
\draw[orange,line width=1.2pt] (O)--(A);

\fill[blue,opacity=0.1] (O)--(A)--(B)--cycle;
\draw[blue] (-2.8,2.8) node[right] {\small \colorbox{blue!2!white}{$I_{P, Q}(s_3)$}};
\draw[blue,line width=1.2pt] (O)--(B) node[above left] { {slope $-1$}};

\fill[orange,opacity=0.1] (O)--(C)--(D)--cycle;
\draw[orange] (-1.6, -3) node {\small \colorbox{orange!2!white}{$I_{P, Q}(s_5)$}};
\draw[orange,line width=1.2pt] (O)--(D);

\draw[Green,line width=1.2pt] (O)--(C);
\fill[Green,opacity=0.1] (O)--(B)--(-3.2,-3.4)--(C)--cycle;
\draw[Green] (-2.24, -0.6) node {\small \colorbox{Green!2!white}{\textcolor{Green}{$I_{P, Q}(s_4)$}}};

\fill[blue,opacity=0.1] (O)--(D)--(E)--cycle;
\draw[blue] (3,-3) node[left] {\small \colorbox{blue!2!white}{$I_{P, Q}(s_0)$}};
\draw[blue,line width=1.2pt] (O)--(E);

\draw[Green,line width=1.2pt] (O)--(F) node[above right] {\textcolor{Green}{slope $\frac54$}};
\fill[Green,opacity=0.1] (O)--(E)--(3.2,3.2)--(F)--cycle;
\draw[Green] (2.3, 0.5) node {\small \colorbox{Green!2!white}{\textcolor{Green}{$I_{P, Q}(s_1)$}}};

\draw (-0.2, -0.2) node[left] {O};
\draw[line width=1.3] (O) circle[radius=0.04];
\end{tikzpicture}
\caption{The set of initial values $(f (0), f (1))$ of $\bigl( \frac{x+2}{x+1},  - \frac{x+3}{x+1} \bigr)$-holonomic sequences $f$ having each of the ultimate signs in \eqref{equation: ultimate sign example elliptic rational}.}
\label{figure: example elliptic rational}
\end{center}
\end{figure}
\end{example}

Note that the solution \eqref{eq:Baku_Koki} is a normal form of a hypergeometric sequence in the sense of \cite{Teg24} and can be found algorithmically. 

The lengths of the ultimate signs in \cref{cor:Spq} and \cref{thm:main} are unbounded. However, restricting $P$ and $Q$ to rational-coefficient polynomials, we have the following corollary. 

\begin{corollary}\label{cor:}
Suppose that $P$, $Q \in \Q(x)$ have no zeros or poles in $\N$. 
Then every $(P, Q)$-holonomic sequence 
has the shortest ultimate sign of length $1$, $2$, $3$, $4$, $6$, $8$ or $12$, 
if it has an ultimate sign at all. 
\end{corollary}

\begin{proof}
We may assume that $P$ has the ultimate sign $(+)$, 
as mentioned immediately before \cref{cor:Spq}. 
Although ultimate signs of length 3 do not appear in the following proof, $(P, Q)$-holonomic sequences can have them when $P$ has the ultimate sign $(-)$.

Of the four cases in \cref{cor:Spq}, 
the only one that does not immediately imply our claim 
is when 
$(P, Q)$ is of $\frac{k}r$-$O$ elliptic type for some coprime positive integers $r$ and $k$.
If $(r, k) = (2, 1)$, we are done. 
Otherwise, $-\frac1{4 \cos^2 \frac{k}r \pi} = \lim\limits_{x \to \infty} \frac{Q(x)}{P(x)P(x-1)} \in \Q$. 
Since $\cos^2 \frac{k}r \pi = \frac12 \left( \cos \frac{2k}r \pi +1 \right)$, 
we have $\cos \frac{2k}r \pi \in \Q$, 
and thus $\cos \frac{2}r \pi \in \Q$ since $r$ and $k$ are coprime. 
The corollary follows from the fact that  
the only possibilities for such $r$ are $3$, $4$, $6$, 
since $f$ will then have the shortest ultimate sign of length $2r \in \{ 6, 8, 12\}$
by \cref{cor:Spq}. (Note that $k=1$ for these cases.)
This fact is known as (a version of) Niven's theorem, 
but we present its proof for the sake of completeness. 

If $r$ were a multiple of $8$, then $\cos \bigl( \frac2r \pi \cdot \frac{r}8 \bigr) = \frac1{\sqrt{2}}$ would be rational, which is a contradiction. 
Thus there is $j \in \{ 0, 1, 2 \}$ such that $2^{-j} r$ is odd. 
Since $\cos \frac2r \pi$ is rational, so is $\cos \frac{2^{j+1}}r \pi$. 
The \emph{Chebyshev polynomial} $T \in \Z[x]$ of order $2^{-j} r$ is the polynomial such that $T(\cos \theta) = \cos 2^{-j} r \theta$ for any $\theta \in \R$, whose leading coefficient and constant term are known to be a non-negative power of 2 and $0$ respectively. It follows from $T\bigl( \cos \frac{2^{j+1}}r \pi \bigr) - 1 = 0$ that $\bigl| \cos \frac{2^{j+1}}r \pi \bigr|$ is a non-positive power of 2. One can get $r = 2, 3, 4, 6$ by some calculation using $\frac12 < \cos \frac{2^{j+1}}r \pi$ when $r$ is large. Since we have an assumption of $r\neq2$, the proof is done. 
\end{proof}

We will derive from \cref{cor:Spq} another corollary (\cref{cor:hyperbolic subsequence} in \cref{sec:proof_of_thm:procedure}). Appropriate subsequences of second-order holonomic sequences are again second-order holonomic sequences. That corollary describes the types of the coefficients of the recurrence which the subsequences satisfy.

\subsubsection{Connection to continued fractions}\label{sec:cfrac}

In this section, we discuss the connection between \cref{thm:main} and convergence theorems of continued fractions
\[
\Kettenbruch _{k=0}^n \frac{Q(k)}{P(k)} 
=
\cfrac{Q(0)}{
P(0) 
+
\cfrac{Q(1)}
{
P(1) + 
\genfrac{}{}{0pt}{0}{\vphantom{l}}{
\lower1ex\hbox{$\ddots$} 
\genfrac{}{}{0pt}{0}{}{
{} + \cfrac{Q(n)}{P(n)}
}
}
}
}.
\]
Note that continued fractions take values in $\hat{\R} = \R \cup \{ \infty \}$ with $x / \infty = 0$ for $x \in \R$ and $x / 0 = \infty$ for $x \in \R \setminus \{ 0 \}$. See \cite{LW08} about their deep theory and application. 
Continued fractions are closely related to second-order holonomic sequences through the next proposition, which can be verified by induction on $n$ (simultaneously for all $P$ and $Q$): 

\begin{proposition}\label{prop:cfrac-holonomic}
Let $P, Q \in \R(x)$ have no poles in $\N$ and $A$ and $B$ be the $(P, Q)$-holonomic sequences with initial values $(1, 0)$ and $(0, 1)$ respectively. Then  
\begin{equation}\label{eq:K=A/B}
\Kettenbruch _{k=0}^n \frac{Q(k)}{P(k)} = \frac{A(n+2)}{B(n+2)}
\end{equation}
in $\hat{\R}$ for all $n \in \N$.
\end{proposition}

For this reason, $A(n)$ and $B(n)$ are called the $n$th canonical numerator and denominator, respectively. We can interpret \cref{thm:main} to a convergence theorem of subsequences $\{ p(A(n), B(n)) \}_{n \equiv i \pmod{\tau}}$, $i=0, \dots, \tau-1$, of $p(A(n), B(n))$ where $p$ is the projection \eqref{eq:projection} and $\tau \geq 1$ is a suitable integer below. 

Let $\tau$ be $2$, $1$, $1$, $2r$ in \cref{thm:main}~\eqref{item:loxodromic_nearly_infty}, \eqref{item:loxodromic_far_from_infty}, \eqref{item:hyperbolic}, \eqref{item:k/r-O_elliptic}, respectively. Then the set $I_i(+)$ of initial values of $(P, Q)$-holonomic sequence $f$ such that $\{ f(n) \}_{n \equiv i \pmod{\tau}}$ has the ultimate sign $(+)$ is a half-plane on $\R^2$. 
Since $f$ satisfies
\begin{equation}\label{eq:f=f0w}
f(n) = A(n)f(0) + B(n)f(1) = \sqrt{A(n)^2 + B(n)^2} \ p(A(n), B(n)) \cdot (f(0), f(1)), 
\end{equation}
$\{ p(A(n), B(n)) \}_{n \equiv i \pmod{\tau}}$ converges to the midpoint of the interval $p(I_i(+))$. 
Similarly, it can be derived that, for any $\tau \geq 1$, one of $\{ p(A(n), B(n)) \}_{n \equiv i \pmod{\tau}}$ must diverge in the case of \cref{thm:main}~\eqref{item:Q-Omega_elliptic}. 
In this sense, \cref{thm:main} is a convergence theorem of the subsequences of $p(A(n), B(n))$. 

By the discussion above, $\bigl\{ -\Kettenbruch _{k=0}^{n} \frac{Q(k)}{P(k)} \bigr\}_{n \equiv i \pmod{\tau}} = \left\{ -\frac{A(n)}{B(n)} \right\}_{n \equiv i \pmod{\tau}}$, $i=0, \dots, \tau-1$ converge to the slope values of the critical lines in the situation \eqref{item:loxodromic_nearly_infty}, \eqref{item:loxodromic_far_from_infty}, \eqref{item:hyperbolic}, \eqref{item:k/r-O_elliptic}. 

\begin{theorem}\label{thm:subcfrac_converge}
Let $P$, $Q \in \R(x)$ be rational functions without zeros or poles in $\N$. First, in \eqref{item:loxodromic_nearly_infty}, \eqref{item:loxodromic_far_from_infty}, \eqref{item:hyperbolic} and \eqref{item:k/r-O_elliptic} of \cref{thm:main}, the slopes of the critical lines are exactly the accumulation points of the continued fraction $\bigl\{ -\Kettenbruch _{k=0}^{n} \frac{Q(k)}{P(k)} \bigr\}_{n \in \N}$. Second, the accumulation of the continued fraction is as follows:
\begin{enumerate}[(1)]
\item \label{item:cfrac_O_lox}
If $(P, Q)$ is of $\infty$-$O$ loxodromic type, 
the subsequences
$\bigl\{ \Kettenbruch _{k=0}^{n} \frac{Q(k)}{P(k)} \bigr\}_{n \equiv i \pmod{2}}$, $i = 0$, $1$, converge in $\hat{\R}$ to 
distinct values. 

\item \label{item:cfrac_Omega_lox}
If $(P, Q)$ is of $\infty$-$\Omega$ loxodromic or hyperbolic type, 
the sequence $\bigl\{ \Kettenbruch _{k=0}^{n} \frac{Q(k)}{P(k)} \bigr\}_{n\in\N}$ converges in $\hat{\R}$. 

\item \label{item:cfrac_O_ell}
If $(P, Q)$ is of $\frac{k}r$-$O$ elliptic type, where $r$ and $k$ are coprime positive integers, 
the sequences $\bigl\{ \Kettenbruch _{k=0}^{n} \frac{Q(k)}{P(k)} \bigr\} _{n \equiv i \pmod{r}}$, $i = 0$, \ldots, $r-1$, converge in $\hat{\R}$ to distinct values.  

\item \label{item:cfrac_Omega_ell}
If $(P, Q)$ is of $\Q$-$\Omega$ elliptic type, 
then for no positive integer $\tau$ and no $i \in \{ 0, \dots, \tau - 1 \}$ does the sequence $\bigl\{ \Kettenbruch _{k=0}^{n} \frac{Q(k)}{P(k)} \bigr\}_{n \equiv i \pmod{\tau}}$ converge in $\hat{\R}$. 
\end{enumerate}
\end{theorem}

We consider the ``gap-$r$ subsequences'' $\bigl\{ \Kettenbruch _{k=0}^{n} \frac{Q(k)}{P(k)} \bigr\} _{n \equiv i \pmod{r}}$ instead of the gap-$2r$ subsequences in \eqref{item:cfrac_O_ell} 
because the limit of $\{ p(A(n), B(n)) \}_{n \equiv i \pmod{2r}}$ is equal to the limit of $\{ p(A(n), B(n)) \}_{n \equiv i+\tau \pmod{2r}}$ except for multiplication by $\pm1$.

Part \eqref{item:cfrac_O_lox} of this theorem is included in \cite[Theorems~3.12 and 3.13]{LW08}. 
Part \eqref{item:cfrac_O_ell} is similar to \cite[Lemma~4.28]{LW08}. 
Part \eqref{item:cfrac_Omega_lox} can be derived from the following well-known convergence theorem. Although Parts \eqref{item:cfrac_O_lox}, \eqref{item:cfrac_Omega_lox} and \eqref{item:cfrac_O_ell} follow from \cref{thm:main}, Part \eqref{item:cfrac_Omega_ell} does not follow from \cref{thm:main} alone since it states divergence instead of convergence. 
We prove \eqref{item:cfrac_Omega_ell} in \cref{sec:proof_of_cfrac} using the convergence theorem below and \cref{cor:hyperbolic subsequence}~\eqref{item:Q-Omega_elliptic->Q-Omega_elliptic} (in \cref{sec:proof_of_thm:procedure}), which is derived from \cref{thm:main}.

\begin{theorem}[{\cite[Theorem~7.1]{Koo91}}] \label{thm:cfrac_cconverge}
Let $P$, $Q \in \R(x)$ be rational functions without zeros or poles in $\N$. 
The continued fraction $\bigl\{ \Kettenbruch _{k=0}^n \frac{Q(k)}{P(k)} \bigr\}_{n \in \N}$ converges in $\hat{\R}$ 
if and only if $(P, Q)$ is of $\infty$-$\Omega$ loxodromic or hyperbolic type. 
\end{theorem}

\subsubsection{Connection to monotonic convergence of continued fractions}\label{sec:monotonic_convergence_thm}

If we identify the ultimate sign of $B$, we can extend the convergence of subsequences of $\frac{A(n)}{B(n)}$ to that of $p(A(n), B(n))$. 
But this is not enough to prove each part of \cref{thm:main}; we need monotonic convergence theorems. 
This is because \cref{thm:main} even describes the ultimate signs of holonomic sequences with initial values on the critical lines, and therefore figures out not only the convergence of subsequences of $p(A(n), B(n))$, but also the direction in which the subsequences of $p(A(n), B(n))$ converge to their limits. 

\cite[Theorems 3.12 and 3.13]{LW08} and \cite[Lemma 3.4]{KKL+21} are monotonic convergence theorems for $(P, Q)$ of $\infty$-O, -$\Omega$ loxodromic type and of hyperbolic type, respectively, and both literature identify the ultimate sign of $B$ in their cases. Hence \cref{thm:main}~\eqref{item:loxodromic_nearly_infty} and \eqref{item:loxodromic_far_from_infty} can be derived from the former literature, and \eqref{item:hyperbolic} can be derived from the latter.

\subsection{Computing the ultimate sign}
\label{subsection: main result algorithm}

The partial algorithm in the following theorem tells us, for given $(P, Q) \in \Q(x)^2$ and $f_0 \in \Q^2$, the index $N \in \N$ at which the $(P, Q)$-holonomic sequence with initial value $f_0$, whenever it terminates. 
Note that once we get $N$, 
we can obtain the ultimate sign itself
by looking at the signs of 
a finite number of terms $f (N)$, $f (N + 1)$, $\dots$
according to \cref{cor:Spq}. 

\begin{theorem}\label{thm:procedure}
There exists a partial algorithm that, 
\begin{itemize}
\item 
given 
$P, Q \in \Q(x)$ without zeros or poles in $\N$, 
together with a pair $f_0 \in \Q^2$,
\item 
terminates if and only if 
the $(P, Q)$-holonomic sequence $f$ with initial value $f_0$
has an ultimate sign
and it is stable 
in the sense that there is a neighbourhood $\mathcal{N} \subseteq \Q^2$ of $f_0$
such that all $(P, Q)$-holonomic sequences with initial value in $\mathcal{N}$ have the same ultimate sign, 
and
\item 
whenever it terminates, 
outputs an index at which $f$ has its ultimate sign.
\end{itemize}
\end{theorem}

Note that the type of $(P, Q)$ can be computed from $P$ and $Q$, 
and hence, 
although the partial algorithm does not terminate
when $f_0 = (0, 0)$ or when $(P, Q)$ is $\Q$-$\Omega$ elliptic
(because of \cref{thm:main}~\eqref{item:Q-Omega_elliptic}), 
we could make it terminate also on these inputs
and declare the non-existence of an ultimate sign in the latter case. 

This partial algorithm terminates on ``most'' inputs since, for $(P, Q)$ of $\infty$-$O$, -$\Omega$ loxodromic, hyperbolic and $\theta$-$O$ type, the $(P, Q)$-holonomic sequence $f$ with initial value $f_0$ has an unstable ultimate sign if and only if $f_0$ is on the finitely many critical lines delimiting the areas $I_{P, Q}(s)$ in \cref{thm:main}. 
For a small but substantial class of $(P, Q)$, it is known that all $f_0 \in \Q^2 \setminus \{(0, 0) \}$ lead $f$ to a stable ultimate sign, or in other words, the slopes of the critical lines are irrational, which is the main topic of \cref{sec:total_algorithm}. 
However there is no known general method to determine the stability, and it is a wide-open problem whether we can make the algorithm terminate on all inputs \cite{KKL+21,IS24,NOW21}.

Theorem~\ref{thm:procedure} is stated for  
rational-coefficient $P$, $Q$ and rational-valued $f _0$, 
so that the problem is computationally meaningful. 
By studying the proofs in some detail 
we could, however, modify the statement appropriately 
so that the partial algorithm accepts inputs involving real numbers
represented as infinite sequences of approximations, 
in a way analogous to the discussion in \cite{Neu21}
about signs of C-finite sequences. 

\begin{example} \label{ex:ct + our pa}
By combining our partial algorithm in \cref{thm:procedure} with creative telescoping, we can determine all values of $n \in \N$ for which an inequality of the form $\sum_{k} F(n,k) > \sum_k G(n,k)$ holds, for some proper hypergeometric terms $F(n,k)$ and $G(n, k)$.
We use creative telescoping to find a holonomic recurrence satisfied by $f(n) := \sum_{k} (F(n,k)-G(n,k))$. If, fortunately, this recurrence is a second-order one, we apply our partial algorithm. If it halts successfully, we can determine all $n$ for which $\sum_{k} F(n,k) > \sum_k G(n,k)$ holds. 
We execute this process for some examples below.

First, we consider the examples related to the sequence $L_t$, which consists of values of Legendre polynomials and was considered in \cref{ex:unknown_slope}.
We will determine all $n$ for which
\begin{equation} \label{eq:L1/2}
\sum_{0 \leq k \leq n/2, \ k \in 2 \Z} \binom{2n-2k}{n-k} \binom{n-k}{k} 
2^{2k-n}
\ >
\sum_{0 \leq k \leq n/2, \ k \in 2 \Z+1} \binom{2n-2k}{n-k} \binom{n-k}{k} 
2^{2k-n}
\end{equation}
holds and $n$ for which
\begin{equation} \label{eq:L1/sqrt2}
\sum_{0 \leq k \leq n/2, \ k \in 2 \Z} \binom{2n-2k}{n-k} \binom{n-k}{k} 
\sqrt{2}^{2k-n}
\ >
\sum_{0 \leq k \leq n/2, \ k \in 2 \Z+1} \binom{2n-2k}{n-k} \binom{n-k}{k} 
\sqrt{2}^{2k-n}
\end{equation}
holds, respectively. These inequalities are equivalent to $L_{1/2}(n)>0$ and $L_{1/\sqrt2}(n)>0$, respectively.
Creative telescoping tells us that the sequence $L_t$, $t \in \R$ is $(P_t, Q)$-holonomic, where $P_t$ and $Q$ are as in \cref{ex:unknown_slope}.
Since $L_t$ turned out to be a second-order holonomic sequence, we can apply our partial algorithm. 
For $t=\frac12, \frac1{\sqrt2}$, it halts successfully, and shows that $L_{\frac12}$ has the ultimate sign $(+,+,-,-,-,+)$ at $0$ and that $L_{\frac1{\sqrt{2}}}$ has the ultimate sign $(+,+,+,-,-,-,-,+)$ at $0$. 
Thus, the inequality \eqref{eq:L1/2} holds for $n \equiv 0, 1, 5 \pmod{6}$ while its reverse inequality holds for $n \equiv 2, 3, 4 \pmod{6}$, and the inequality \eqref{eq:L1/sqrt2} holds for $n \equiv 0, 1, 2 , 7 \pmod{8}$ while the reverse inequality holds for $n \equiv 3,4,5,6 \pmod{8}$.

Secondly, we will determine all $n$ for which
\[
\sum_{0 \leq k \leq n, \ k \in 2\Z} k \binom{n}{k}^3
> 
\sum_{0 \leq k \leq n, \ k \in 2\Z+1} k \binom{n}{k}^3
\]
holds. To do this, we want to find the ultimate sign of the difference $f(n) := \sum_{0 \leq k \leq n} (-1)^k k \binom{n}{k}^3$ and when $f$ has it. 
Creative telescoping tells us that $f$ is a 
\[
(R(x), S(x)) := \left( \frac{18x^2+36x+12}{(x+1)(x+2)(6x^2+4x+1)}, - \frac{3(3x+2)(3x+1)(6x^2+16x+11)}{(x+1)(x+2)(6x^2+4x+1)} \right)\text{-}
\]
holonomic sequence. 
Since $f$ turned out to be a second-order holonomic sequence, we can apply our partial algorithm. It halts successfully, and shows that $f$ has the ultimate sign $(+, -, -, +)$ at $1$.
Therefore, for $n\geq 1$, the above inequality holds if $n \equiv 0, 3 \pmod{4}$ while its reverse inequality holds if $n \equiv 1, 2 \pmod{4}$.
\end{example}

\begin{example} \label{ex:partial algorithm}
Our partial algorithm in \cref{thm:procedure} enables us to numerically approximate the slopes of the critical lines (see the paragraphs below \cref{thm:main} for the definition) to arbitrary precision, since when it receives an input $(P, Q, f_0) \in \Q(x)^2 \times \Q^2$ with $(P, Q)$ of loxodromic, hyperbolic, or $\theta$-$O$ elliptic type, it halts if and only if $f_0$ does not belong to the critical lines. 
For example, in \cref{figure: example t=1/2 and t=1/sqrt2} of \cref{ex:unknown_slope}, we can approximate the slopes of critical lines as 
\begin{align*}
-7.2875 &< \text{the slope of $l_0$} < -7.2873, 
&
0.02832 < \text{the slope of $l_1$} < 0.02833,
\\
0.68821 &< \text{the slope of $l_2$} < 0.68822,
\end{align*}
in the left picture, and as
\begin{align*}
-3.627 &< \text{the slope of $l_0$} < -3.6269,
&0.054373 &< \text{the slope of $l_1$} < 0.054374, 
\\
0.493139 &< \text{the slope of $l_2$} < 0.49314, 
&1.0564 &< \text{the slope of $l_3$} < 1.0565,
\end{align*}
in the right picture.
\end{example}

\cref{thm:procedure} can be described in a reduction form that is an extension of \cite[Theorem~3.1]{KKL+21}: 

\begin{theorem}\label{thm:problem_reduction}
    For second-order holonomic sequences, the Ultimate Sign Problem Turing-reduces to the Minimality Problem. 
\end{theorem}

\subsection{Input set admitting a total algorithm}\label{sec:total_algorithm}

\cref{thm:NOW21} gives a sufficient condition on $P, Q \in \Q(x)$ 
for all non-zero rational $(P, Q)$-holonomic sequences $f \in \Q^{\N} \setminus \{ 0 \}$ to have 
stable ultimate signs. This gives a nontrivial input set on which the Ultimate Sign Problem is solvable by the partial algorithm in \cref{thm:procedure}. 

The main predecessor to our work \cite[Theorem 1, 3 and 7]{NOW21} 
relies on \cite[Lemma 14]{NOW21} whose proof contained an error in the calculation of an inverse image. 
Their classification and the partial algorithm \cite[Theorem 1 and 3]{NOW21} 
analogous to our \cref{thm:main,thm:procedure}
are correct after all, as our theorems imply. 
\cite[Theorem 7]{NOW21} is revised into the following \cref{thm:NOW21} with a straightforward generalization and a slight restriction. 
The generalization lies in relaxing the condition from $P, Q \in \Z[x]$ and $0 \notin P(\N), Q(\N)$ to $P, Q \in \Q[x]$ and $P(\N), Q(\N) \subseteq \Z \setminus {0}$.
The restriction is that the condition $s q_1 - p_1 - s < 3 p_0$ (for $d = 1$) and $s q_1 - p_1 < (d+2) p_0$ (for $d \geq 2$) is replaced by a slightly stronger condition, as described in the latter part of the assumptions listed under case \eqref{item:|q|=p}.

\begin{theorem}\label{thm:NOW21}
Let $P(x) = p_0 x^d + p_1 x^{d-1} + \dots + p_d \in \Q [x]$ and $Q(x) = q_0 x^d + q_1 x^{d-1} + \dots + q_d \in \Q [x]$ take nonzero integer values, i.e., $P(\N), Q(\N) \subseteq \Z \setminus \{ 0 \}$.  
Suppose that $p_0 > 0$ and $d \geq 1$ (where $q_0$ might be zero). 
Then, if $P$ and $Q$ satisfy either of the following conditions, any $(P, Q)$-holonomic sequence $f \in \Q^{\N} \setminus \{ 0 \}$ has a stable ultimate sign. 
\begin{enumerate}[(1)]
\item\label{item:|q|<p} $|q_0| < p_0$
\item\label{item:|q|=p} $|q_0| = p_0$ and the two conditions below hold for $s := \sgn q_0 \in \{ 1, -1 \}$:
\begin{itemize}
\item $Q(x) - sP(x) \neq 1$ in $\Q[x]$, 
\item $\begin{cases}
sq_1 - p_1 - s < p_0 & \text{if $d=1$}, \\
sq_1 - p_1 < p_0 & \text{if $d \geq 2$}.
\end{cases}
$
\end{itemize}
\end{enumerate}
\end{theorem}

When $Q(x) - sP(x) \neq 1$ does not hold in the above theorem, the following proposition (with $\lambda = \pm 1$) implies that there exists a rational holonomic sequence with an unstable ultimate sign. 
This proposition is obtained by rephrasing \cite[Proposition~4]{NOW21} and eliminating unnecessary assumptions from it.
We prove this in \cref{sec:recovery_of_NOW21}.

\begin{proposition}\label{prop:unstable}
Let $P \in \R(x)$ have no poles or zeros in $\N$, have ultimate sign $(+)$, and have degree at least 1. Let $\lambda \in \R \setminus \{ 0 \}$. 
Then the $(P, \lambda P + \lambda^2)$-holonomic sequence $\{ (-\lambda)^n \}_{n \in \N}$ has an unstable ultimate sign. 
\end{proposition}

\section{Proof of the Main Results}\label{sec:proofs}

In this section, we prove \cref{thm:main,thm:procedure,thm:problem_reduction}. 
All the proofs of the lemmas in the following \cref{sec:proof_of_thm:main,sec:proof_of_thm:procedure} are postponed to \cref{sec:proof_of_lemmas}. 

\subsection{Proof of \cref{thm:main}} \label{sec:proof_of_thm:main}

Let us first focus on identifying the lengths of the ultimate signs that $(P, Q)$-holonomic sequences can have and get an overview of the proof of \cref{thm:main}. 
\cref{lem:(+),lem:(+-)} below, by types of $(P, Q)$, characterize $(P, Q)$ admitting $(P, Q)$-holonomic sequences with ultimate signs of lengths $1$ and $2$, respectively. Then only lengths $\tau \geq 3$ are left. For each $\tau \geq 3$, we will introduce a special recurrence such that we can decide if $F \in \R^{\N}$ satisfying the recurrence has a (shortest) ultimate sign of length $\tau$ (\cref{lem:single-term-feedback_recurrence}). 
Next, by types of $(P, Q)$, we characterize $(P, Q)$ and $\tau$ that allow all $(P, Q)$-holonomic sequences $f$ to be transformed to $F$ satisfying the special recurrence and having the same ultimate sign as $f$ (\cref{lem:transformation_property}). Finally we show that, for the other $(P, Q)$ and $\tau \geq 3$, no non-zero $(P, Q)$-holonomic sequences have the shortest ultimate sign of length $\tau$ in the proof of \cref{thm:main}~\eqref{item:Q-Omega_elliptic}. 
Note that some lemmas below are superfluous for identifying the lengths of ultimate signs, but required to identify the ultimate signs themselves and how they partition the space of the initial values.

\begin{lemma} \label{lem:(+)}
    Let $P, Q \in \R(x)$ have no zeros or poles in $\N$ and $P$ have the ultimate sign $(+)$. 
    \begin{enumerate}[(1)]
        \item \label{item:condition_of_exist(+)}
        $I_{P, Q}(+) \neq \varnothing$
        $\iff$
        $(P, Q)$ is of loxodromic type or hyperbolic type.
        \item \label{item:(PQ):hyp->+or-}
        If $(P, Q)$ is of hyperbolic type, then $I_{P, Q}(+) \cup I_{P, Q}(-) = \R^2 \setminus \{(0, 0)\}$. 
    \end{enumerate}
\end{lemma}

Similar results to the above lemma appear in, e.g., \cite{KKL+21}.

The following lemma is relatively easy and similar propositions appear in context of continued fractions (e.g., \cite[Theorem~3.12]{LW08}). 

\begin{lemma} \label{lem:(+-)}
Let $P, Q \in \R(x)$ have no zeros or poles in $\N$ and $P$ have the ultimate sign $(+)$. 
\begin{enumerate}[(1)]
\item \label{item:exist(+-)}
$I_{P, Q}(+, -) \neq \varnothing$ $\iff$ $(P, Q)$ is of loxodromic type. 

\item \label{item:pI(+-)_is_closed}
$p(I_{P, Q}(+, -))$ is a closed interval. 

\item \label{item:(PQ):lox->+or-or+-or-+}
If $(P, Q)$ is of loxodromic type, then $I_{P, Q}(+) \cup I_{P, Q}(-) \cup I_{P, Q}(+, -) \cup I_{P, Q}(-, +) = \R^2 \setminus \{ (0, 0) \}$. 
\end{enumerate}
\end{lemma}

Now we introduce the special recurrence mentioned in the first paragraph of this section. 
For a (not necessarily holonomic) sequence $F \in \R^{\N}$, consider a 
\emph{single-term-feedback recurrence} 
    \begin{equation} \label{eq:single-term-feedback_recurrence}
    F(n+\tau) - F(n) = R(n) F(n+1), 
    \end{equation}
where $\tau$ is an integer $\geq 2$ and $R \in \R^{\N}$.
This recurrence expresses the difference between two neighbouring terms in the gap-$\tau$ subsequences $\{ F(n) \}_{n \equiv i \pmod{\tau}}$, $i = 0, \dots, \tau-1$, as a single term in the next subsequence $\{ F(n) \}_{n \equiv i+1 \pmod{\tau}}$ multiplied by the coefficient $R$. 
In the following lemma, we treat the case where $|R(n)|$ rapidly decreases in \eqref{item:R=O} and the case where $|R(n)|$ does not rapidly decrease in \eqref{item:R=Omega}.

\begin{lemma} \label{lem:single-term-feedback_recurrence}
Let $F \in \R^{\N}$ satisfy the single-term-feedback recurrence~\eqref{eq:single-term-feedback_recurrence} for a coefficient $R \in \R^{\N}$ and an integer $\tau \geq 2$. 
\begin{enumerate}[(1)]
\item \label{item:R=O}
(restricted case of \cite[Theorem~6]{Koo91})
Suppose $R(n) = O(n^{-1-\varepsilon})$ for some $\varepsilon > 0$. 
\begin{enumerate}[(\ref{item:R=O}a)]
\item \label{item:F_converge}
Each of the gap-$\tau$ subsequences $\{ F(n) \}_{n \equiv i \pmod{\tau}}$, $i = 0, \dots, \tau-1$, converges. 
\item \label{item:limFneq0} 
If $F \neq 0$, then there is $i \in \{ 0, \dots, \tau-1 \}$
for which $\{ F(n) \}_{n \equiv i \pmod{\tau}}$ does not converge to $0$.  
\end{enumerate}

\item \label{item:R=Omega} 
Suppose that $|R(n)| = \Omega(n^{-1})$ and $R$ has an ultimate sign $(+)$ or $(-)$. If $F$ has an ultimate sign of length $\tau$, then $F$ also has an ultimate sign of length $\leq 2$. 

\item \label{item:zero_convergence_sign}
Suppose that $R$ has an ultimate sign $(q)$, $q \in \{ +, -, 0 \}$. 
Let $i \in \{ 0, \dots, \tau-1 \}$. 
If a subsequence $\{ F(n) \}_{n \equiv i+1 \pmod{\tau}}$ of $F$
has the ultimate sign $(s)$, $s \in \{ +, -, 0 \}$
and 
$\{ F(n) \}_{n \equiv i \pmod{\tau}}$ converges to $0$, then $\{ F(n) \}_{n \equiv i \pmod{\tau}}$ has the ultimate sign $(-qs)$. 
\end{enumerate}
\end{lemma}

In the situation of \eqref{item:R=O}, $F$ has an ultimate sign of length $\tau$ as follows. 
If $F=0$, it is obvious. If $F \neq 0$, then 
by (\ref{item:F_converge}) and \eqref{item:limFneq0}, 
there is $i$ such that 
$\{ F(n) \}_{n \equiv i \pmod{\tau}}$ has 
the ultimate sign $(+)$ or $(-)$. 
Then $\{ F(n) \}_{n \equiv i-1 \pmod{\tau}}$ also has 
$(+)$ or $(-)$ if it converges to a non-zero real number. 
It has $(+)$, $(-)$ or $(0)$ even if it converges to zero by \eqref{item:zero_convergence_sign}. 
Thus, by induction, 
every gap-$\tau$ subsequence of $F$ has ultimate sign of length $1$, 
meaning that $F$ has an ultimate sign of length $\tau$. 
On the other hand, in the situation of \eqref{item:R=Omega}, $F$ does not have the shortest ultimate sign of length $\tau \geq 3$. 

Part (\ref{item:R=O}) of Lemma~\ref{lem:single-term-feedback_recurrence}
is known
for a larger class of recurrences \cite[Theorem~6]{Koo91}. 
Our restriction to the single-term-feedback recurrence allows \eqref{item:R=Omega} and \eqref{item:zero_convergence_sign} to hold. 

Now we want to find sequences $T, R \in \R ^\N$
such that for each $(P, Q)$-holonomic sequence $f$, 
the transformed sequence $F(n) := T(n)f(n)$
has the same ultimate sign as $f$ and
satisfies the recurrence \eqref{eq:single-term-feedback_recurrence}. $F$ and $f$ have the same ultimate sign if and only if $T$ has the ultimate sign $(+)$. 
To find the condition on $T$ and $R$ for $F$ to satisfy the recurrence \eqref{eq:single-term-feedback_recurrence},
we use $A^{(\tau)}, B^{(\tau)} \in \R(x)$ below. 

\begin{definition} \label{def:generalized_normal_numerator/denominator}
For $P, Q \in \R(x)$ without zeros or poles in $\N$, there uniquely exist $A^{(\tau)}, B^{(\tau)} \in \R(x)$ such that any $(P, Q)$-holonomic sequence $f$ satisfies the recurrence 
\begin{equation}\label{eq:f=Bf+Af}
f(n+\tau) = B^{(\tau)}(n) f(n+1) + A^{(\tau)}(n) f(n)
\end{equation}
for all $n \in \N$. Let us call $A^{(\tau)}$ and $B^{(\tau)}$ the \emph{generalized $\tau$th canonical numerator} and \emph{denominator} (of $(P, Q)$) respectively. 
\end{definition}

These are generalizations of the notions of 
$\tau$th canonical numerator $A$ and denominator $B$ in \cref{prop:cfrac-holonomic} since $(A^{(\tau)}(0), B^{(\tau)}(0)) = (A(\tau), B(\tau))$. We can generalize \cref{eq:K=A/B} to $\Kettenbruch_{k=n}^{n+\tau} \frac{Q(k)}{P(k)} = \frac{A^{(\tau+2)}(n)}{B^{(\tau+2)}(n)}$. 
\cref{eq:f=Bf+Af} is a generalization of the equation
$f(\tau) = B(\tau)f(1) + A(\tau)f(0)$ 
that $A$ and $B$ satisfy 
for any $(P, Q)$-holonomic sequence $f$.

Let $\tau \geq 2$ and $T, R \in \R^{\N}$. For each $n \in \N$, by \cref{eq:f=Bf+Af}, $F(n) = T(n)f(n)$ satisfy \cref{eq:single-term-feedback_recurrence} for all $(P, Q)$-holonomic sequences $f$ if and only if 
\begin{equation} \label{eq:condition_of_TandR}
T(n+\tau) A^{(\tau)}(n) = T(n), \quad R(n)T(n+1) = B^{(\tau)}(n) T(n+\tau).
\end{equation} 

To allow $T$ to have the ultimate sign $(+)$, 
we want $A^{(\tau)}$ to have $(+)$. 
In addition, to apply \cref{lem:single-term-feedback_recurrence}~\eqref{item:R=O} for $F(n) = T(n)f(n)$, the absolute value of the coefficient $|R(n)|$ has to decrease rapidly. The next lemma shows that there exists $\tau$ satisfying these conditions if and only if $(P, Q)$ is of $O$ type. 

\begin{lemma} \label{lem:transformation_property}
Let $P, Q \in \R(x)$ have no zeros or poles in $\N$, and $P$ have the ultimate sign $(+)$. Let $\tau \geq 2$ be an integer and $A^{(\tau)}$ and $B^{(\tau)}$ be the $\tau$th generalized canonical numerator and denominator, respectively. 
\begin{enumerate}[(1)]
\item \label{item:condition_of_TandR} 
Assume that $T, R \in \R^{\N}$ satisfy \eqref{eq:condition_of_TandR} and $T(n) \neq 0$ for all sufficiently large $n$. Then $\left| \frac{T(n+1)}{T(n)} \right| = \Theta(|A^{(\tau)}(n)|^{-1 / \tau})$. Especially, $|R(n)| = \Theta ( |B^{(\tau)}(n)||A^{(\tau)}(n)|^{-1+1/\tau} )$. 

\item \label{item:A>0andB/A=n^-1-epsilon_iff_Otype_and_tautheta_in_2Z}
The following are equivalent. 
\begin{enumerate}[(\ref{item:A>0andB/A=n^-1-epsilon_iff_Otype_and_tautheta_in_2Z}a)]
\item \label{item:A>0andB/A=n^-1-epsilon}
$A^{(\tau)}$ has the ultimate sign $(+)$ and $|B^{(\tau)}(n)|  A^{(\tau)} (n)^{-1+1/\tau} = O(n^{-1-\varepsilon})$ for some $\varepsilon > 0$. 

\item \label{item:Otype_and_tautheta_in_2Z}
$(P, Q)$ is of $\theta$-$O$ elliptic type and $\tau \theta \in 2 \Z$, or $(P, Q)$ is of $\infty$-$O$ loxodromic type and $\tau \in 2 \Z$. 
\end{enumerate}
\end{enumerate}
\end{lemma}

Now we are ready to show \cref{thm:main}. 

\begin{proof}[Proof of \cref{thm:main}~\eqref{item:loxodromic_nearly_infty} and \eqref{item:loxodromic_far_from_infty}]
By \cref{lem:(+-)}~\eqref{item:(PQ):lox->+or-or+-or-+}, it remains to prove that $p(I_{P, Q}(+, -))$ has width if $(P, Q)$ is of $\infty$-$O$ loxodromic type and does not if $(P, Q)$ is of $\infty$-$\Omega$ loxodromic type. In other words, we should prove the existence of a $(P, Q)$-holonomic sequence with the stable ultimate sign $(+, -)$ in the former case and the non-existence in the latter case. 

Define $T, R \in \R^{\N}$ as they satisfy $T(n), R(n) > 0$ and the relation \eqref{eq:condition_of_TandR} for $\tau = 2$ for all sufficiently large $n$. (Note that $A^{(2)} = Q$ and $B^{(2)} = P$.) Then, for all $(P, Q)$-holonomic sequences $f$ and all sufficiently large $n$, the transformed sequences $F(n) := T(n)f(n)$ satisfy the single-term-feedback recurrence \eqref{eq:single-term-feedback_recurrence} for $\tau = 2$, i.e.,
\begin{equation} \label{eq:single-term-feedback_recurrence_with_tau=2}
F(n+2) - F(n) = R(n) F(n+1).
\end{equation}
Since $A^{(2)}(n) = Q(n) > 0$ for all sufficiently large $n$, \cref{lem:transformation_property} implies $R(n) = O(n^{-1-\varepsilon})$ for some $\varepsilon > 0$ if $(P, Q)$ is of $\infty$-$O$ loxodromic type and $R(n) = \Omega(n^{-1})$ if $(P, Q)$ is of $\infty$-$\Omega$ loxodromic type. 

If $(P, Q)$ is of $\infty$-$O$ loxodromic type and so $R(n) = O(n^{-1-\varepsilon})$, we can define a linear map $L$ that maps a $(P, Q)$-holonomic sequence $f$ to 
\[
L(f) := \left( \lim\limits_{\substack{n \equiv 0 \pmod{2}, \\ n \to \infty}} T(n)f(n), \quad \lim\limits_{\substack{n \equiv 1 \pmod{2}, \\ n \to \infty}} T(n)f(n) \right) \in \R^2
\]
by \cref{lem:single-term-feedback_recurrence}~\eqref{item:F_converge}. 
By \cref{lem:single-term-feedback_recurrence}~\eqref{item:limFneq0}, $L$ is injective. Since the domain and range of $L$ are both two-dimensional, $L$ is bijective. Hence, for example, $L^{-1}(1, -1)$ is a $(P, Q)$-holonomic sequence that has the stable ultimate sign $(+, -)$. 

If $(P, Q)$ is of $\infty$-$\Omega$ loxodromic type, take a $(P, Q)$-holonomic sequence $f$ with the ultimate sign $(+, -)$. 
Let us show that this is unstable. It suffices to show that $T(n)f(n) = O(1)$ and $\lim\limits_{n \to \infty}T(n)g(n) = \infty$ where $g$ is a $(P, Q)$-holonomic sequence with the ultimate sign $(+)$ (because it follows that, for any $\delta > 0$, the perturbations $f + \delta g$ of $f$ have the ultimate sign $(+)$). For all sufficiently large $n$, since $R(n) > 0$ and $F(n) = T(n)f(n)$ satisfies \cref{eq:single-term-feedback_recurrence_with_tau=2}, $F(2n) \ (> 0)$ is monotonically decreasing and $F(2n+1) \ (< 0)$ is monotonically increasing. 
So $F(n) = O(1)$. On the other hand, $F'(n) := T(n)g(n) \ (>0)$, a sequence satisfying the same recurrence, eventually increasing. Especially $F'(n) = \Omega(1)$. Since $R(n) = \Omega(n^{-1})$, we have $F'(n+2) - F'(n) = \Omega(n^{-1})$. 
Thus $\lim\limits_{n \to \infty} F'(n) = \infty$. 
\end{proof}

\begin{proof}[Proof of \cref{thm:main}~\eqref{item:hyperbolic}]
$I_{P, Q}(+)$, $I_{P, Q}(-)$ are both connected and $I_{P, Q}(+) = -I_{P, Q}(-)$. The statement follows from this and \cref{lem:(+)}~\eqref{item:(PQ):hyp->+or-}. 
\end{proof}

\begin{proof}[Proof of \cref{thm:main}~\eqref{item:Q-Omega_elliptic}]
Suppose, for a contradiction, that a non-zero $(P, Q)$-holonomic sequence $f$ has an ultimate sign $(s_0, \dots, s_{\tau-1})$. 

Let $\tau \geq 3$ first. Let $A^{(\tau)}$ and $B^{(\tau)}$ be the generalized $\tau$th canonical numerator and denominator. It follows from \cref{lem:transformation_property}~\eqref{item:A>0andB/A=n^-1-epsilon_iff_Otype_and_tautheta_in_2Z} that $A^{(\tau)}$ has the ultimate sign $(-)$ or $(0)$, or that $A^{(\tau)}$ has $(+)$ and $|B^{(\tau)}(n)|  A^{(\tau)} (n)^{-1+1/\tau} = \Omega(n^{-1})$. 
Let us first consider the former case. 
Let $(b) \ (b \in \{+, -, 0 \})$ be the ultimate sign of $B^{(\tau)}$. 
Comparing the signs of the three terms in \cref{eq:f=Bf+Af}, we have $s_i = b s_{i+1}$ for all $i = 0, \dots, \tau-1$, where $s_{\tau} := s_0$, and so $f$ has an ultimate sign of length $\leq 2$. 
Next, let us consider the latter case. 
We can choose $T, R \in \R^{\N}$ satisfying $T(n) > 0$ and the relation~\eqref{eq:condition_of_TandR} for all sufficiently large $n$. Then we have $|R(n)| = \Omega(n^{-1})$. 
The transformed sequence $F(n) := T(n)f(n)$ satisfies the recurrence~\eqref{eq:single-term-feedback_recurrence} for all sufficiently large $n$. 
It follows from \cref{lem:single-term-feedback_recurrence}~\eqref{item:R=Omega} that $F$ has an ultimate sign of length $\leq 2$, and so does $f$. 

Now it remains to consider the case $\tau = 1, 2$. By \cref{lem:(+)}~\eqref{item:condition_of_exist(+)} and \cref{lem:(+-)}~\eqref{item:exist(+-)}, $f$ does not have ultimate signs of length $1$ or $2$. 
\end{proof}

It remains to show \eqref{item:k/r-O_elliptic}. Let $(P, Q)$ be of $\theta$-$O$ elliptic type. As already mentioned, for $\tau$ such that $\tau \theta \in 2 \Z$, all $(P, Q)$-holonomic sequences $f$ have ultimate signs of length $\tau$. Now we need to determine which ultimate signs (of length $\tau$) $f$ can have. This will be derived from the following lemma. 

\begin{lemma} \label{lem:Tg^j_converge}
Take $(P, Q)$ as in \cref{lem:transformation_property} and assume that it is of $\frac{k}r$-$O$ elliptic type. 
\begin{enumerate}[(1)]
\item \label{item:sgnR=sgnq}
The generalized $2r$th canonical denominator $B^{(2r)}$ has the ultimate sign $(+)$, $(-)$ and $(0)$ if $\frac{Q(x)}{P(x)P(x-1)}$ is eventually increasing, if it is eventually decreasing and if it is constant, respectively. 

\item \label{item:Tf^j}
By \cref{lem:transformation_property}~\eqref{item:A>0andB/A=n^-1-epsilon_iff_Otype_and_tautheta_in_2Z}, we can choose $T \in \R^{\N}$ such that $T(n) > 0$ and the relation~\eqref{eq:condition_of_TandR} for $\tau = 2r$ hold for all sufficiently large $n$. Then, for each $j = 0, \dots, \tau-1$, there exists a $(P, Q)$-holonomic $f^{(j)}$ such that for each $i \in \{0, \dots, \tau-1\}$, the $i$th subsequence $\{ T(n)f^{(j)}(n) \}_{n \equiv i \pmod{2r}}$ converges to a real number of sign $\sgn \sin \frac{j - ik}r \pi$.
\end{enumerate}
\end{lemma}

\begin{proof}[Proof of \cref{thm:main}~\eqref{item:k/r-O_elliptic}]
Take $T$ and $f^{(0)}, \dots, f^{(2r-1)}$ as in \cref{lem:Tg^j_converge}~\eqref{item:Tf^j}. Let $f^{(2r)} := f^{(0)}$. $(P, Q)$-holonomic sequences of the form $f = a f^{(j)} + b f^{(j+1)} \ (a, b > 0)$ have the ultimate sign $s_j$ since each $\{ T(n)f(n) \}_{n \equiv i \pmod{2r}}$, $i=0, \dots, 2r-1$, converges to a real number of sign $\sgn \sin \frac{j - ik + 0.5}r \pi$. Then we have $\{ \text{initial values of } a f^{(j)} +  b f^{(j+1)}  \mid a, b > 0 \} \subseteq I_{P, Q}(s_{j}) $. It remains to prove that $f^{(j)}$ has the ultimate sign $s_{j}$, $s_{j-1}$ and $t_j$ if $\frac{Q(x)}{P(x)P(x-1)}$ is eventually increasing, if it is eventually decreasing and if it is constant, respectively. 

For $i, j \in \{ 0, \dots, 2r-1 \}$ and $q \in \{ 0, \pm 1\}$, let $u_{i, j, q} := \sgn \sin \frac{j-ik+q/2}r \pi$. Then what we want to prove is that $f^{(j)}$ has the ultimate sign $(u_{i, j, q})_{i=0, \dots, 2r-1}$, where $(\sgn q)$ is the ultimate sign of $B^{(2r)}$ in \cref{lem:Tg^j_converge}~\eqref{item:sgnR=sgnq}. We will show that the subsequence $\{ T(n)f^{(j)}(n) \}_{n \equiv i \pmod{2r}}$ has the ultimate sign $u_{i, j, q}$ for each $i$. 

If $j - ik \not\equiv 0 \pmod{r}$, then this subsequence converges to a real number of sign $u_{i, j, 0} \ (\neq 0)$. Therefore it has the ultimate sign $(u_{i, j, 0}) = (u_{i, j, q})$. If $j - ik \equiv 0 \pmod{r}$, then this subsequence converges to $0$. Define $R \in \R^{\N}$ by the relation~\eqref{eq:condition_of_TandR}. $R$ has the ultimate sign $(\sgn q)$ and $F(n) = T(n)f^{(j)}(n)$ satisfies \eqref{eq:single-term-feedback_recurrence}. It follows from \cref{lem:single-term-feedback_recurrence}~\eqref{item:zero_convergence_sign} that this subsequence has the ultimate sign $(-\sgn q u_{i+1, j, 0}) = (\sgn q (-1)^{\frac{j-ik}r} ) = (u_{i, j, q})$. 
\end{proof}

\subsection{Proof of \cref{thm:procedure,thm:problem_reduction}} \label{sec:proof_of_thm:procedure}

\cref{thm:procedure,thm:problem_reduction} are algorithmic claims stating that
the ultimate signs can be partially computed in each sense. 
We could prove them by analyzing the proof of Theorem~\ref{thm:main} quantitatively. 
But instead of carrying out such analysis for each case of Theorem~\ref{thm:main} separately, 
we choose to do so just for the hyperbolic type (Lemma~\ref{lem:quantified(+)} below),
and argue that all other types (having ultimate signs) reduce to it
in the sense of the following \cref{cor:hyperbolic subsequence}.

From the original recurrence \eqref{eq:(PQ)holonomic}, 
we can obtain, for each positive integer $\tau$, a ``gap-$\tau$ recurrence'' 
\begin{equation} \label{eq:f2tau=Pftau+Qf}
f(n+2\tau) = P_{\tau}(n) f(n+\tau) + Q_{\tau}(n) f(n), 
\end{equation}
where $P _\tau$ and $Q _\tau$ are rational functions.
Specifically, they can be written as
\begin{align} \label{eq:Ptau_and_Qtau}
P_{\tau} & = \frac{B^{(2\tau)}}{B^{(\tau)}}, 
&
Q_{\tau} & = A^{(2\tau)} - \frac{B^{(2\tau)}}{B^{(\tau)}} A^{(\tau)}
\end{align}
using the generalized canonical numerators $A ^{(0)}$, $A ^{(1)}$, $\dots$ and denominators $B ^{(0)}$, $B ^{(1)}$, $\dots$ of $(P, Q)$
(see Definition~\ref{def:generalized_normal_numerator/denominator}),
assuming that $B ^{(\tau)}$ is non-zero. 
(Note that if $B ^{(\tau)} = 0$, we have $f (n + \tau) = A ^{(\tau)} (n) f (n)$,
in which case the ultimate sign of $f$ can be found easily.)
Thus,
the subsequence $\{ f(\tau n + N ) \}_{n \in \N}$ of $f$, 
for any number $N \in \N$ greater than all zeros of $B ^{(\tau)}$, 
is the $(P_{\tau}(\tau x + N), Q_{\tau}(\tau x + N))$-holonomic sequence 
with initial value $(f(N), f(N + \tau)) $.
The following corollary to Theorem~\ref{thm:main} says that
with a right choice of $\tau$,
this $(P_{\tau}(\tau x + N), Q_{\tau}(\tau x + N))$ is of hyperbolic type, 
unless $(P, Q)$ is of $\Q$-$\Omega$ elliptic type. 

\begin{corollary}
\label{cor:hyperbolic subsequence}
Suppose that $P$, $Q \in \R (x)$ have no zeros or poles in $\N$.
Let $A^{(0)}$, $A^{(1)}$, $\dots$ and $B^{(0)}$, $B^{(1)}$, $\dots$ be the generalized canonical numerators and denominators, respectively. 
\begin{enumerate}[(1)]
\item \label{item:lox_or_k/r-O_elliptic->hyepr} 
Suppose that 
$(P, Q)$ is either of loxodromic type
or of $\frac k r$-$O$ elliptic type for some coprime positive integers $r$ and $k$.
Let $\tau = 2$ in the former case,
and $\tau = 2 r$ in the latter case. 
Suppose that $B ^{(\tau)}$ and $B ^{(2 \tau)}$ are non-zero. 
Then 
$P _\tau$ and $Q _\tau$ defined by \eqref{eq:Ptau_and_Qtau} are non-zero, 
and $(P _\tau (\tau x + N), Q _\tau (\tau x + N))$ is of hyperbolic type for all $N \in \N$.
\item \label{item:Q-Omega_elliptic->Q-Omega_elliptic} 
Suppose that
$(P, Q)$ is of $\Q$-$\Omega$ elliptic type.
Then $B ^{(\tau)}$ is non-zero, $P _\tau$ and $Q _\tau$ defined by \eqref{eq:Ptau_and_Qtau} are also non-zero,
and 
$(P _\tau (\tau x + N), Q _\tau (\tau x + N))$ is of $\Q$-$\Omega$ elliptic type for all $N \in \N$ and $\tau \geq 1$.
\end{enumerate}
\end{corollary}

\begin{proof}
\eqref{item:lox_or_k/r-O_elliptic->hyepr} 
$P_{\tau}, Q_{\tau} \neq 0$ follows from $B^{(\tau)}, B^{(2\tau)} \neq 0$. Since the type of $(P_{\tau}(\tau x + N), Q_{\tau}(\tau x + N))$ does not depend on $N$, it suffices to prove this corollary only for $N$ which is larger than any zero and pole of $P_{\tau}$, $Q_{\tau}$, $B^{(\tau)}$. Since $\{ (f(N), f(N + \tau)) \mid \text{$f$ is a $(P, Q)$-holonomic sequence} \} = \R^2$ by $B^{(\tau)}(N) \neq 0$, when $f$ runs on the set of all $(P, Q)$-holonomic sequences, $\{ f(\tau n + N ) \}_{n \in \N}$ runs on the set of all $(P_{\tau}(\tau x + N), Q_{\tau}(\tau x + N))$-holonomic sequences. By \cref{cor:Spq}, any $\{ f(\tau n + N ) \}_{n \in \N}$ has an ultimate sign $(+)$, $(-)$, or $(0)$. Hence, again by \cref{cor:Spq}, $(P_{\tau}(\tau x + N), Q_{\tau}(\tau x + N))$ is of hyperbolic type. 

\eqref{item:Q-Omega_elliptic->Q-Omega_elliptic} 
By \cref{cor:Spq},
no non-zero $(P, Q)$-holonomic sequence has an ultimate sign. 
Therefore $B^{(\tau)} \neq 0$ for any $\tau$. 
Since the type of $(P_{\tau}(\tau x + N), Q_{\tau}(\tau x + N))$ does not depend on $N$, it suffices to prove this corollary for one $N$. Non-zero $(P, Q)$-holonomic sequences do not have ultimate signs, so there exists at least one $N \in \N$ such that the subsequence $\{ f( \tau n + N) \}_{n \in \N}$ does not have any ultimate signs. Therefore $P_{\tau}, Q_{\tau} \neq 0$, and it follows from \cref{cor:Spq} that $(P_{\tau}(\tau x + N), Q_{\tau}(\tau x + N))$ is of $\Q$-$\Omega$ elliptic type. 
\end{proof}

\begin{lemma}[A quantitative version of \cref{lem:(+)}] \label{lem:quantified(+)}
Let $P, Q \in \R(x)$ have no zeros or poles in $\N$.
\begin{enumerate}[(1)]
    \item \label{item:(PQ):lox_or_hyp_iff_q(1-q)>Q}
    The following are equivalent.
    \begin{enumerate}[(\ref{item:condition_of_exist(+)}a)]
            \item \label{item:(PQ):lox_or_hyp}
            $(P, Q)$ is of loxodromic or hyperbolic type. 
            
            \item \label{item:q(1-q)>Q}
            There exists $q \in \R^{\N}$ with ultimate sign $(+)$ that satisfies 
                \begin{equation}\label{eq:q(1-q)>Q}
                q(n) (1-q(n+1)) \geq -\frac{Q(n)}{P(n)P(n-1)}
                \end{equation}
            for all sufficiently large $n \in \N$.
        \end{enumerate}

\item \label{item:q=1/2+1/4n+1/4nlogn}
If \eqref{item:q(1-q)>Q} holds, 
then it holds 
for the sequence $q$ defined by 
$q(0) = q(1) = 1$ and
$q(n) = \frac12 + \frac1{4n} + \frac1{4n\log n}$, $n \geq 2$.

\item \label{item:f/f>q}
Let $(P, Q)$ be of hyperbolic type and $P$ have the ultimate sign $(+)$. Take any $q$ in \eqref{item:q(1-q)>Q}. Take $N \in \N$ such that $P$, $q$, $Q$ have their ultimate signs at $N$ and the condition~\eqref{eq:q(1-q)>Q} is satisfied for any $n \geq N$. Let $f$ be a $(P, Q)$-holonomic sequence. Then if
\begin{equation} \label{eq:f/f>q}
f(n) \neq 0 \text{ and } \frac{f(n+1)}{f(n)} > q(n)P(n-1)
\end{equation}
holds for some $n \geq N$, this condition also holds for $n+1, n+2, \dots$. In particular, $f$ has an ultimate sign $(+)$ or $(-)$ at $n$. 
\end{enumerate}
\end{lemma}

The sequence $q$ in \cref{lem:quantified(+)}~\eqref{item:q=1/2+1/4n+1/4nlogn} is what appears in the proof of \cite[Lemma~3.4]{KKL+21}.

\begin{proof}[Proof of Theorem~\ref{thm:procedure}]
The desired partial algorithm simply diverges when the input $(P, Q)$ is of $\Q$-$\Omega$ elliptic type. 
For the input $(P, Q)$ of hyperbolic type together with $f_0 \in \Q^2$, define $q \in \R^{\N}$ as in \cref{lem:quantified(+)}~\eqref{item:q=1/2+1/4n+1/4nlogn} and execute the following procedure:

\begin{enumerate}
\item 
If $P$ has the ultimate sign $(-)$, then write $f_0 = (a, b)$, and let $P := -P$ and $f_0 := (a, -b)$.

\item \label{item:initialN}
Calculate any $N$ as in \cref{lem:quantified(+)}~\eqref{item:f/f>q}. 

\item 
Let $f$ be the $(P, Q)$-holonomic sequence with initial value $f_0$. For $n = N, N+1, \dots$, check the condition \eqref{eq:f/f>q}, and if it is satisfied then output $n$. 
\end{enumerate}

Let us show that if this procedure halts, then the output is correct and the $(P, Q)$-holonomic sequence $f$ with initial value $f_0$ has a stable ultimate sign. Without loss of generality, we can assume that $P$ has the ultimate sign $(+)$. 
It follows from \cref{lem:quantified(+)}~\eqref{item:f/f>q} that $f$ has an ultimate sign at the output $n$ when the procedure halts. Moreover, since $\sgn f(n)$ and the condition \eqref{eq:f/f>q} are robust under small perturbations of the initial value of $f$, the ultimate sign of $f$ is stable. 

Conversely, let us assume that the $(P, Q)$-holonomic sequence $f$ with initial value $f_0$ has a stable ultimate sign. 
By \cref{lem:(+)}~\eqref{item:(PQ):hyp->+or-}, $f$ has $(+)$ or $(-)$. Without loss of generality, we can assume it is $(+)$. 
Let $N$ be the number obtained in step~\ref{item:initialN} of the procedure with input $P$, $Q$, $f_0$. 
It follows from the stability of the ultimate sign of $f$ that there exists a $(P, Q)$-holonomic sequence $g$ such that 
\begin{itemize}
\item $g(N) > 0$,
\item $g$ satisfies the condition \eqref{eq:f/f>q} for $n=N$, where $f$ is replaced by $g$,
\item A small perturbation $f-g$ of (the initial value of) $f$ has the same ultimate sign $(+)$ as $f$. 
\end{itemize}
We want to show that $F(n) := g(n) / \prod_{k=N}^{n-1}q(k)P(k-1) \to \infty \ (n \to \infty)$ since then we have $\lim\limits_{n \to \infty} f(n) / \prod_{k=N}^{n-1}q(k)P(k-1) = \infty$ and the condition \eqref{eq:f/f>q} holds for some $n$. By the assumption of $g$ and \cref{lem:quantified(+)}~\eqref{item:f/f>q}, $g$ (and so $F$) has the ultimate sign $(+)$ at $N$. Recurrence $g(n+2) = P(n)g(n+1) + Q(n)g(n)$ and the condition \eqref{eq:q(1-q)>Q} yield that $F(n+2) - F(n+1) \geq (q(n+1)^{-1} - 1) (F(n+1) - F(n))$ for all $n \geq N$. Note that $F(N+1) - F(N) > 0$. Then we have $F(n+2) - F(n+1) = \Omega \left( \prod_{k=0}^{n} (q(k+1)^{-1} - 1) \right) $. 
Since $q(k+1)^{-1} - 1 = 1 - \frac1k - \frac1{k \log k} + O(k^{-2})$, it follows that $\prod_{k=0}^{n} (q(k+1)^{-1} - 1) = \Theta\left( \frac1{n \log n} \right)$. (Herein we used $\prod_{k=2}^n (1 + \frac{\alpha}k + \frac{\beta}{k \log k}) = \Theta(n^{\alpha}(\log n)^{\beta})$ for arbitrary $\alpha, \beta \in \R$.) Thus $F(n+2) - F(n+1) = \Omega\left( \frac1{n \log n} \right)$, and so $F(n) = \Omega(\log\log n)$, which proves $F(n) \to \infty$. 

Finally, when the input $(P, Q)$ is of loxodromic type or $\frac{k}r$-$O$ elliptic type, define $\tau$ as in \cref{cor:hyperbolic subsequence}. If the $\tau$th generalized canonical denominator $B^{(\tau)}$ or the $2\tau$th one $B^{(2\tau)}$ is $0$, it is easy to make our partial algorithm behave as in \cref{thm:procedure}. 
Assume that $B^{(\tau)}, B^{(2\tau)} \neq 0$, and define $P_{\tau}$, $Q_{\tau}$ by \eqref{eq:Ptau_and_Qtau}. Let $N_0 \in \N$ be larger than any pole of $P_{\tau}$ and $Q_{\tau}$. Since all $(P_{\tau}(\tau x + N), Q_{\tau}(\tau x + N))$ for $N=N_0, \dots, N_0 + \tau -1$ are of hyperbolic type, we can execute the aforementioned procedure with inputs $(P_{\tau}(\tau x + N), Q_{\tau}(\tau x + N))$ and $(f(N), f(N+\tau))$ for each $N$,
which each halts if and only if 
$\{ f( \tau n + N ) \}_{n \in \N}$ has a stable ultimate sign of length $1$.
All of these $\tau$ executions thus halt if and only if $f$ has a stable ultimate sign of length $\tau$.
\end{proof}

\begin{proof}[Proof of \cref{thm:problem_reduction}]
    Take inputs $f_0 \in \Q^2$ and $(P, Q) \in \Q(x)^2$ for the Ultimate Sign Problem for second-order holonomic sequences. Let $f$ be the $(P, Q)$-holonomic sequence with initial value $f_0$. Without loss of generality, we can assume that $P$, $Q$ are non-zero, $(P, Q)$ is not of $\Q$-$\Omega$ type and $f_0 \neq (0, 0)$ (otherwise the problem is easy). 
    We can also assume that $P$, $Q$ have no zeros in $\N$.
    As in the proof of \cref{thm:procedure}, we only have to consider the case of $(P, Q)$ of hyperbolic type, by taking a suitable subsequence. 

    Assume that one has an oracle for the Minimality Problem for second-order holonomic sequences. This oracle tells us whether $f$ has an unstable ultimate sign, since it is equivalent to the minimality of $f$ for $(P, Q)$ of hyperbolic type. 
    
    If $f$ has a stable ultimate sign, execute the partial algorithm in \cref{thm:procedure}. Otherwise, take $q$ as in \cref{lem:quantified(+)}~\eqref{item:q=1/2+1/4n+1/4nlogn}, and calculate and output $N$ of \eqref{item:f/f>q} in the same lemma. Let us show that this output is correct. If $f(n) = 0$ for some $n \geq N$, then $f(n+1) \neq 0$ and $\frac{f(n+2)}{f(n+1)} = P(n) > q(n+1) P(n)$, which is the condition \eqref{eq:f/f>q} for $n+1$. This implies that $f$ has a stable ultimate sign at $n+1$, which is a contradiction. If $f$ has no zeros $\geq N$ and satisfies $\frac{f(n+1)}{f(n)} < 0$ for some $n \geq N$, then 
    \begin{equation} \label{eq:f/f>q_holds}
        \frac{f(n+2)}{f(n+1)} = P(n) + Q(n) \frac{f(n)}{f(n+1)}  
        \geq
        P(n)
        >
        q(n+1) P(n),
    \end{equation}
    resulting in the same as above. Thus, $\frac{f(n+1)}{f(n)} > 0$ for all $n \geq N$. 
\end{proof}

\subsection{Proof of the lemmas} \label{sec:proof_of_lemmas}

\begin{proof}[Proof of \cref{lem:quantified(+)}]
\eqref{item:(PQ):lox_or_hyp}$\implies$\eqref{item:q(1-q)>Q} and \eqref{item:q=1/2+1/4n+1/4nlogn} These follow from the inequality
\[
q(n)(1-q(n+1)) \geq \frac14 + \frac{1}{16n^2} + \frac{1}{16n^2(\log n)^2} - \frac1{n^3}
\]
for all $n \geq 3$, where $q$ is defined as in \eqref{item:q=1/2+1/4n+1/4nlogn}. (You can show this inequality using $n^{-1} - \frac12 n^{-2} \leq \log(1+n^{-1}) \leq n^{-1}$. )

\eqref{item:q(1-q)>Q}$\implies$\eqref{item:(PQ):lox_or_hyp} Suppose, for a contradiction, that \eqref{item:q(1-q)>Q} holds and $(P, Q)$ is of elliptic type. Take $q$ in \eqref{item:q(1-q)>Q}. Then there exists $C > \frac1{16}$ such that for all sufficiently large $n$, 
\begin{equation} \label{eq:q(1-q)>1/4}
q(n) (1-q(n+1)) > \frac14 + \frac{C}{n^{2}}. 
\end{equation}
Especially, we have $0 < q(n) < 1$ for all sufficiently large $n$. If $q(n) < q(n+1)$, then $q(n)(1-q(n+1)) < q(n+1)(1-q(n+1)) \leq \frac14$, which contradicts the equation above. Hence $q$ is eventually decreasing, and $\alpha := \lim\limits_{n \to \infty} q(n) \geq 0$ exists. Letting $n \to \infty$ in \cref{eq:q(1-q)>1/4} gives $\alpha(1-\alpha) \geq \frac14$, so $\alpha = \frac12$. Define $p(n)$ so that $q(n) = \frac12 + p(n) / n$. Then $p$ has the ultimate sign $(+)$, and by the inequality~\eqref{eq:q(1-q)>1/4} we have
\begin{equation} \label{eq:1/2(np-np)-pp>C/n^2}
\frac{n}2 (p(n) - p(n+1)) + \frac12 p(n) - p(n)p(n+1) > \frac{C}{n^{2}} (n+1) n > C. 
\end{equation}
If $p(n) < p(n+1)$, we have the left-hand side of the above inequality $\leq \frac12 p(n) - p(n)^2 \leq \frac1{16}$, which contradicts $C > \frac1{16}$. Therefore $p(n)$ is eventually decreasing, and $\beta := \lim\limits_{n \to \infty} p(n) \in \R$ exists. If $\liminf\limits_{n \to \infty} n (p(n) - p(n+1)) > 0$, then we have $p(n) - p(n+1) = \Omega(n^{-1})$ and so $\lim\limits_{n \to \infty} p(n) = -\infty$, which contradicts the existence of $\beta$. Thus $\liminf\limits_{n \to \infty} n (p(n) - p(n+1)) \leq 0$. Taking $\liminf\limits_{n \to \infty}$ of inequality~\eqref{eq:1/2(np-np)-pp>C/n^2} yields $\frac12 \beta - \beta^2 \geq C > \frac1{16}$. This is a contradiction.

\eqref{item:f/f>q} Use the condition~\eqref{eq:f/f>q} and the inequality~\eqref{eq:q(1-q)>Q} to obtain
\[
\frac{f(n+2)}{f(n+1)} 
=
P(n) + Q(n) \frac{f(n)}{f(n+1)}
>
P(n) + \frac{Q(n)}{q(n)P(n-1)}
\geq 
q(n+1)P(n).
\]
Then the assertion follows by induction. 
\end{proof}

\begin{proof}[Proof of \cref{lem:(+)}]
\eqref{item:condition_of_exist(+)} It suffices to prove $I_{P, Q}(+) \neq \varnothing$ $\iff$ \cref{lem:quantified(+)}~\eqref{item:q(1-q)>Q}. 
$I_{P, Q}(+) \neq \varnothing$, i.e., there exists a $(P, Q)$-holonomic sequence $f$ with the ultimate sign $(+)$, if and only if there exists $q \in \R^{\N}$ with the ultimate sign $(+)$ such that, for all sufficiently large $n$,
\[
 q (n + 1) = 1 - \frac{-Q (n)}{P (n) P (n - 1)} \cdot \frac{1}{q (n)}, 
\]
which is an equation obtained from the recurrence~\eqref{eq:(PQ)holonomic} by rewriting with $q (n) = f (n + 1) / (f (n)P (n - 1))$. 
Since the right-hand side monotonically increases as $q(n) (> 0)$ increases, the existence of such $q$ is equivalent to the existence of $q$ that satisfies 
\[
 0 < q (n + 1) \leq 1 - \frac{-Q (n)}{P (n) P (n - 1)} \cdot \frac{1}{q (n)}, 
\]
an inequality version of the above equation, for all sufficiently large $n$. This is \eqref{item:q(1-q)>Q}. 

\eqref{item:(PQ):hyp->+or-} Take $q$ and $N$ as in \cref{lem:quantified(+)}~\eqref{item:f/f>q}. We want to prove that $\sgn f(n)$, $n \geq \max \{ N, 1 \}$, changes at most once for any non-zero $(P, Q)$-holonomic sequence $f$. Assume that $\sgn f(n)$ changes into $\sgn f(n+1) \neq 0$ at $n \geq \max \{ N, 1 \}$. Then, by a similar discussion in the proof of \cref{thm:problem_reduction}, the inequality \eqref{eq:f/f>q_holds} holds. 
By \cref{lem:quantified(+)}~\eqref{item:f/f>q}, the sign of $f$ does not change after $n+1$. 
\end{proof}

\begin{proof}[Proof of \cref{lem:(+-)}]
Assume first that $(P, Q)$ is not of loxodromic type. If a $(P,Q)$-holonomic sequence $f$ has the ultimate sign $(+, -)$, then we find a contradiction by comparing the signs of the three terms of the equation $f(n+2) = P(n)f(n+1) + Q(n)f(n)$. Hence $I_{P, Q}(+, -) = \varnothing$, which proves \eqref{item:exist(+-)} and \eqref{item:pI(+-)_is_closed}. There is nothing to prove for \eqref{item:(PQ):lox->+or-or+-or-+}. 

Next, assume that $(P, Q)$ is of loxodromic type. 
Take $N \in \N$ at which $P$, $Q$ have the ultimate sign $(+)$. For any non-zero $(P, Q)$-holonomic sequence $f$, if successive terms $f(n)$, $f(n+1) \ (n \geq N)$ have the same sign, then all the following terms have the same sign. It follows from this that the non-empty closed subset (of $\R^2 \setminus \{ (0, 0) \}$) $I_n = \{ f_0 \in \R^2 \setminus \{(0, 0)\} \mid \text{The $(P, Q)$-holonomic sequence } f \text{ with initial value } f_0 \text{ satisfies } f(2n) \geq 0 \text{ and } f(2n+1) \leq 0. \}$ is decreasing for $n \geq N$. Since $p(I_n)$ are decreasing non-empty closed intervals, $\bigcap_{n \geq N} p(I_n)$ is also non-empty closed intervals. 
This proves \eqref{item:exist(+-)} and \eqref{item:pI(+-)_is_closed}. Finally let us show \eqref{item:(PQ):lox->+or-or+-or-+}. Take a non-zero $(P, Q)$-holonomic sequence $f$ with the initial value $f_0 \notin I_{P, Q} (+, -) \cup I_{P, Q}(-, +) = \left( \bigcap_{n \geq N} I_n \right) \cup \left( \bigcap_{n \geq N} (-I_n) \right)$. Then there exists $n \geq N$ such that $f_0 \notin I_n \cup (-I_n)$, i.e., $\sgn f(2n) = \sgn f(2n+1)$. Thus $f$ has the ultimate sign $(+)$ or $(-)$. 
\end{proof}

\subsubsection*{Proof of \cref{lem:single-term-feedback_recurrence}}

\begin{proof}[Proof of \cref{lem:single-term-feedback_recurrence}~\eqref{item:R=O}]
Set $S_{N, n} := \sum_{k=N}^{n-1} |R(n)|$. There exists $S_{N, \infty} := \lim\limits_{n \to \infty} S_{N, n} \in \R$. By taking large $N$, we can assume $S_{N, \infty} < \frac12$. 
For any $n \in \N$ and $N' \in \{ N, N+1, \dots, N+\tau-1 \}$ such that $n \geq N'$ and $n \equiv N' \pmod{\tau}$, let us prove the following upper bound of the variation of $F$ by course-of-values induction on $n$. 
\begin{equation} \label{eq:|f-f|<max}
\left| F(n) - F (N') \right| 
\leq
2 S_{N, n} \max_{N \leq I < N + \tau} \left| F (I) \right|
\end{equation}
If $n = N'$, it is obvious. Assume $n > N'$. Set $\displaystyle C := \max_{N \leq I < N + \tau} \left| F (I) \right|$. By the induction hypothesis and the recurrence~\eqref{eq:single-term-feedback_recurrence}, we have
\[
\begin{aligned}
\left| F(n) - F (N') \right| 
\leq&
\left| F(n) - F(n-\tau) \right|  + \left| F(n-\tau) - F(N') \right| \\
\leq&\ 
| R(n-\tau) | \left| F(n-\tau+1) \right| + 
2 S_{N, n-\tau} C.
\end{aligned}
\]
Let us find an upper bound of $\left| F(n-\tau+1) \right|$ in this inequality. Consider $N'' \in \{ N, N+1, \dots, N+\tau-1\}$ such that $N'' \equiv n-\tau+1 \pmod{\tau}$. Then the induction hypothesis and $S_{N, \infty} < \frac12$ give a bound:
\[
\left| F (n-\tau+1) \right| 
\leq
\left| F (n-\tau+1) - F (N'') \right| + \left| F (N'') \right| 
\leq 
2 S_{N'', n-\tau+1} C + C
\leq 2C. 
\]
Since $|R(n-\tau)| \leq S_{n-\tau, n}$, we obtain the bound~\eqref{eq:|f-f|<max}. 

For any $N$ such that $S_{N, \infty} < \frac12$ and any $n \in \N$, $N' \in \{ N, N+1, \dots, N+\tau-1 \}$ such that $n \geq N'$ and $n \equiv N' \pmod{\tau}$, by the bound~\eqref{eq:|f-f|<max}, we especially have 
\begin{equation} \label{eq:|f-f|<max_simple}
|F(n) - F(N')| 
\leq 
2 S_{N, \infty} \max_{N \leq I < N+\tau} |F(I)|. 
\end{equation}

\eqref{item:F_converge} By \eqref{eq:|f-f|<max_simple}, $F$ is bounded. Therefore $\left\{ \max\limits_{N \leq I < N+\tau} |F(I)| \right\}_{N \in \N}$ is also bounded. Since $S_{N, \infty} \to 0$ as $N\to \infty$, it follows from the bound~\eqref{eq:|f-f|<max_simple} that each $\{ F(n)\}_{n \equiv i \pmod{\tau}}$, $i=0, \dots, \tau-1 $ is a Cauchy sequence.

\eqref{item:limFneq0} Take $I$ that realizes the ``$\max$'' on the right-hand side of \eqref{eq:|f-f|<max_simple} and set $N' = I$. Then $|F(n) - F(I)| \leq 2 S_{N, \infty} |F(I)|$. Letting $N \to \infty$ with keeping the condition $n \equiv I \pmod{\tau}$ in this inequality, we find that the limit of $\{ F(n) \}_{n \equiv I \pmod{\tau}}$ is not $0$, since $2 S_{N, \infty} < 1$. 
\end{proof}

\begin{proof}[Proof of \cref{lem:single-term-feedback_recurrence}~\eqref{item:zero_convergence_sign}]
The right-hand side of the recurrence~\eqref{eq:single-term-feedback_recurrence} has the sign $qs$ for sufficiently large $n$ with $n \equiv i \pmod{\tau}$. Therefore $\{ F(n) \}_{n \equiv i \pmod{\tau}}$ is eventually increasing if $qs$ is positive, eventually decreasing if negative, and constant if $0$. Thus, it has the ultimate sign $(-qs)$. 
\end{proof}

\begin{proof}[Proof of \cref{lem:single-term-feedback_recurrence}~\eqref{item:R=Omega}]

Since the right-hand side of \cref{eq:single-term-feedback_recurrence} has a constant sign for sufficiently large $n$ with $n \equiv i \pmod{\tau}$, the subsequence $\{ F(n) \}_{n \equiv i \pmod{\tau}}$ is eventually monotonic. Then there exists $L_i := \lim\limits_{\substack{n \equiv i \pmod{\tau} \\ n \to \infty}} F(n) \in \R \cup \{ \pm \infty \}$ for each $i = 0, \dots, \tau-1$. 

Let us show that $L_0 = \dots = L_{\tau-1} = 0$ or $L_0, \dots, L_{\tau-1} \in \{ \pm \infty\}$. We assume $L_i \neq 0$ for some $i$ and show $L_0, \dots, L_{\tau-1} \in \{ \pm \infty\}$. By $R(n) = \Omega(n^{-1})$ and $L_i \neq 0$, the recurrence~\eqref{eq:single-term-feedback_recurrence} yields $F(n+\tau) - F(n) = \Omega(n^{-1})$ where $n \equiv i-1 \pmod{\tau}$. Therefore $L_{i-1} \in \{ \pm \infty \}$. The same discussion proves $L_{i-2}, L_{i-3}, \ldots  \in \{ \pm \infty\}$ by induction on $i$.

By what we showed above and a similar discussion in the proof of \cref{lem:single-term-feedback_recurrence}~\eqref{item:zero_convergence_sign}, for each $i=0, \dots, \tau-1$, the ultimate sign of the subsequence $\{ F(n) \}_{n \equiv i \pmod{\tau}}$ and whether the previous subsequence $\{ F(n) \}_{n \equiv i-1 \pmod{\tau}}$ eventually increases or decreases are determined by whether $\{ F(n) \}_{n \equiv i \pmod{\tau}}$ eventually increases or decreases. 
Therefore, there exists $s \in \{ 0, \pm 1 \}$ such that $s$ is independent of $i$ and the ultimate sign of $\{ F(n) \}_{n \equiv i-1 \pmod{\tau}}$ is that of $\{ F(n) \}_{n \equiv i \pmod{\tau}}$ multiplied by $s$. 
Thus, $F$ has an ultimate sign of length $\leq 2$. 
\end{proof}

\subsubsection*{Proof of \cref{lem:transformation_property,lem:Tg^j_converge}} 

\begin{proof}[Proof of \cref{lem:transformation_property}~\eqref{item:condition_of_TandR}]
There exists $N \in \N$ such that $A^{(\tau)}(n) \neq 0$ for all $n \geq N$ since $T(n) \neq 0$. By $|T(n)| = \Theta \left( \prod\limits_{\substack{k \equiv n \!\!\! \pmod{\tau},\\ N \, \leq \, k \, \leq \, n-\tau}} \dfrac1{|A^{(\tau)}(k)|} \right)$, we have
\[
\left| \frac{T(n+1)}{T(n)} \right|
=
\Theta
\left( |A^{(\tau)}(n)|^{-1 / \tau}
\prod_{\substack{k \equiv n \!\!\! \pmod{\tau},\\ N \, \leq \, k \, \leq \, n-\tau}} \left| 
\frac{A^{(\tau)}(k+\tau)^{\frac1{\tau}} A^{(\tau)}(k)^{1 - \frac1{\tau}}}{A^{(\tau)}(k+1)} 
\right|
\right). 
\]
Each factor $\frac{A^{(\tau)}(k+\tau)^{\frac1{\tau}} A^{(\tau)}(k)^{1 - \frac1{\tau}}}{A^{(\tau)}(k+1)}$ of the product is $1 + O(k^{-2})$, so the product converges as $n \to \infty$. Therefore $\left| \frac{T(n+1)}{T(n)} \right| = \Theta \left( |A^{(\tau)}(n)|^{-1/\tau} \right)$. Especially, 
\[
|R(n)| = \left| B^{(\tau)}(n) \frac{T(n+\tau)}{T(n+\tau-1)} \dotsm \frac{T(n+2)}{T(n+1)} \right| = \Theta(|B^{(\tau)}(n)||A^{(\tau)}(n)|^{1-1 / \tau}). \qedhere
\]
\end{proof}

To prove \cref{lem:transformation_property}~\eqref{item:A>0andB/A=n^-1-epsilon_iff_Otype_and_tautheta_in_2Z} and \cref{lem:Tg^j_converge}, let us study the properties of the generalized $\tau$th canonical numerator and denominator. 

\begin{lemma} \label{lem:generalized_canonical_numerator/denominator}
Let $P, Q \in \R(x)^2$ have no zeros or poles in $\N$. The generalized $i$th canonical denominators $B^{(i)} \in \R(x)$ of $(P, Q)$ satisfy the recurrence
\begin{equation}\label{eq:B}
B^{(i+2)}(x) = P(x) B^{(i+1)}(x+1) + Q(x+1) B^{(i)}(x+2), \quad (B^{(0)}, B^{(1)}) = (0, 1). 
\end{equation}
The generalized $i$th canonical numerator ($i \geq 1$) is $A^{(i)}(x) = Q(x)B^{(i-1)}(x+1)$. 
\end{lemma}

\begin{proof}
For any $(P, Q)$-holonomic sequence $f$, the term $f(n+i+2)$ is expressed by $f(n)$ and $f(n+1)$ as follows: 
\[
\begin{aligned}
f((n+1)+(i+1)&) = B^{(i+1)}(n+1) f(n+2) +A^{(i+1)} (n+1) f(n+1) \\
=& \left( P(n)B^{(i+1)}(n+1)  + A^{(i+1)} (n+1) \right) f(n+1) + Q(n)B^{(i+1)}(n+1) f(n) 
\end{aligned}
\]
Comparing this to \cref{eq:f=Bf+Af} for $\tau = i+2$, we obtain the lemma by induction on $i$. 
\end{proof}

Let us calculate the ultimate sign of $B^{(i)}$ and $\deg B^{(i)}$. Let $\deg 0 := - \infty$. 

\begin{lemma} \label{lem:B^i}
Let $P, Q \in \R(x)$ have no zeros or poles in $\N$ and $P$ have the ultimate sign $(+)$. Let $i \geq 1$ be an integer and $B^{(i)}$ be the generalized $i$th canonical denominator. Let $L := \lim\limits_{x \to \infty} \frac{Q(x)}{P(x)P(x-1)} \in [-\infty, \infty]$. Then $\deg B^{(i)}$ is:
\[
\begin{cases}
    (i-1) \deg P  + \deg \left( \frac{Q(x)}{P(x)P(x-1)} + \frac1{4\cos^2 \theta \pi} \right) & \! \text{if $L = -\frac1{4 \cos^2 \theta \pi} \in (-\infty, -\frac14)$ and $ i \theta \in \Z$},
    \\
    (i-1) \deg P + \lfloor \frac{i-1}2 \rfloor \max\left\{ 0, \deg \frac{Q(x)}{P(x)P(x-1)} \right\} & \! \text{Otherwise}.
\end{cases}
\]
Let $q \in \{ +, -, 0 \}$ be $+$ if $\frac{Q(x)}{P(x)P(x-1)}$ is eventually increasing, $-$ if eventually decreasing, and $0$ if constant. Then the ultimate sign of $B^{(i)}$ is:
\[
\begin{cases}
    (+) &  \text{if $L \geq -\frac14$},
    \\
    (\sgn \sin i \theta) & \text{if $L = -\frac1{4 \cos^2 \theta \pi} \in [-\infty, -\frac14) $ and $ i \theta \notin \Z$},
    \\
    (\sgn (-1)^{i \theta} q) & \text{if $L = -\frac1{4 \cos^2 \theta \pi} \in [-\infty, -\frac14)$ and $ i \theta \in \Z$}.
\end{cases}
\]
\end{lemma}

\begin{proof}
Since $\frac{B^{(i)}(x)}{P(x+i-2) \dotsm P(x)}$ is the generalized $i$th canonical denominator of $\left( 1, \frac{Q(x)}{P(x)P(x-1)} \right)$, we can assume $P=1$ without loss of generality. If $L = \pm \infty$, it follows by induction on $i$ from the recurrence~\eqref{eq:B} for $P=1$ that
$
\lim\limits_{x \to \infty} B^{(i)}(x) / Q (x)^{\lfloor \frac{i-1}2 \rfloor} = 
\begin{cases}
1 & (i \notin 2 \Z) \\
\frac{i}2 & (i \in 2 \Z)
\end{cases}.
$
This proves the lemma in this case. 

If $L \in (-\infty, \infty)$, let $b_i := \lim\limits_{x \to \infty} B^{(i)}(x)$. Letting $x \to \infty$ in the recurrence~\eqref{eq:B} for $P=1$, we have
\begin{equation} \label{eq:b=b+Lb}
b_{i+2} = b_{i+1} + L b_i.
\end{equation}
If $L \in [-\frac14, \infty)$, then $b_i > 0$ for all $i \geq 1$, so the lemma follows. Assume $L \in (-\infty, -\frac14)$ and let $L = -\frac1{4 \cos^2 \theta \pi}$. Then $b_i = \frac{ \sin i \theta \pi}{\sin \theta \pi (2 \cos \theta \pi)^{i-1}}$. This proves the claim in the case where $i \theta \notin \Z$. Moreover, by induction on $i$, it follows from the recurrence~\eqref{eq:B} for $P=1$ that
\[ \label{eq:B^i=sin_i*theta*pi}
\begin{aligned}
B^{(i)} (x) 
&= 
b_i
-
\frac{2\varepsilon(x)}{\tan^2 \theta \pi}
\left( 
\frac{i \cos (i-1) \theta \pi}{(2 \cos \theta \pi)^{i-1} }
-
b_i
\right)
+
O\left( x^{-1} \varepsilon(x) \right),
\\
\varepsilon(x) :&= Q(x) + \frac1{4\cos^2 \theta \pi}.
\end{aligned}
\]
If $i \theta \in \Z$, then the above expression is $B^{(i)}(x) = \frac{2\varepsilon(x)}{\tan^2 \theta \pi} (-1)^{i \theta + 1} \frac{i \cos \theta \pi}{(2 \cos \theta \pi)^{i-1} } + O\left( x^{-1} \varepsilon(x) \right)$. The lemma follows from this and $\sgn \varepsilon (x) = - \sgn q$ for large $x$. 
\end{proof}

\begin{proof}[Proof of \cref{lem:transformation_property}~\eqref{item:A>0andB/A=n^-1-epsilon_iff_Otype_and_tautheta_in_2Z}]
One can verify this lemma using $A^{(\tau)}(x) = Q(x)B^{(\tau-1)}(x+1)$ (by \cref{lem:generalized_canonical_numerator/denominator}) and \cref{lem:B^i}. 
\end{proof}

\begin{proof}[Proof of \cref{lem:Tg^j_converge}]
\eqref{item:sgnR=sgnq} This immediately follows from \cref{lem:B^i}. 

\eqref{item:Tf^j} By \cref{lem:single-term-feedback_recurrence}~\eqref{item:R=O}, there exist linear maps $L_i$, $i=0, \dots, 2r-1$ that map $(P, Q)$-holonomic sequences $f$ to $L_i(f) := \lim\limits_{\substack{n \equiv j \pmod{2r}, \\ n \to \infty}} T(n)f(n)$. Let $j \in \{0, \dots, 2r-1 \}$ and take $j' \in \{ 0, \dots, r-1\}$ such that $j'k \equiv j \pmod{r}$. Since the range of $L_{j'}$ has a lower dimension than its domain, $L_{j'}$ is not injective. Therefore there exists a non-zero $(P, Q)$-holonomic sequence $f^{(j)}$ such that $L_{j'}(f^{(j)}) = 0$. Without loss of generality, we can assume that $\sgn L_{j'+1}(f^{(j)}) \in \{ 0, \sgn (-1)^{\frac{j'k - j}r + 1} \}$. (Otherwise we consider $-f^{(j)}$ instead of $f^{(j)}$.)

Let us first show that $\sgn L_{j' + i} (f^{(j)}) = \sgn L_{j'+1} (f^{(j)}) \sin \frac{ik}r \pi$. Let $A^{(0)}, A^{(1)}, \dots$ and $B^{(0)}, B^{(1)}, \dots$ be the generalized canonical numerators and denominators. Then we have 
\begin{equation}\label{eq:Tf=TTBTf+TTATf}
T(n+i)f^{(j)}(n + i) 
=
\frac{T(n+i)}{T(n+1)} B^{(i)}(n) T(n+1)f^{(j)}(n+1) 
+ \frac{T(n+i)}{T(n)} A^{(i)}(n) T(n)f^{(j)}(n). 
\end{equation}
It suffices to show that the right-hand side converges to a real number whose sign is $\sgn L_{j'+1} (f^{(j)}) \sin \frac{ik}r \pi$, as $n \to \infty$ keeping the condition $n \equiv j' \pmod{2r}$. It follows from \cref{lem:transformation_property}~\eqref{item:condition_of_TandR} and \cref{lem:generalized_canonical_numerator/denominator} that 
\begin{align*}
\frac{T(n+i)}{T(n+1)} \lvert B^{(i)}(n) \rvert &= \Theta\left(\left(Q(n) B^{(2r-1)}(n+1)\right)^{-(i-1) / 2r} B^{(i)}(n)\right), 
\\
\frac{T(n+i)}{T(n)} \lvert A^{(i)}(n) \rvert &= \Theta\left(\left(Q(n) B^{(2r-1)}(n+1)\right)^{-i / 2r} Q(n)B^{(i-1)}(n+1)\right). 
\end{align*}
Then, using \cref{lem:B^i}, one can verify the followings:
\begin{itemize}
    \item $\frac{T(n+i)}{T(n)} A^{(i)}(n) = O(1)$, 
    \item $\frac{T(n+i)}{T(n+1)} \lvert B^{(i)}(n) \rvert = \Theta(1)$ if $i \neq 0, r$, and
    \item $\lim\limits_{n\to\infty} \frac{T(n+i)}{T(n+1)} B^{(i)}(n) = 0$ if $i = 0, r$. 
\end{itemize}
Hence, the right-hand side of the equation \eqref{eq:Tf=TTBTf+TTATf} converges to a real number of the desired sign.
(Note that we used $L_{j'}(f^{(j)}) = 0$ and the ultimate sign of $B^{(i)}$ shown in \cref{lem:B^i}.) 

By \cref{lem:single-term-feedback_recurrence}~\eqref{item:limFneq0}, there exists $i$ such that $L_{j' + i} (f^{(j)}) \neq 0$. Since $\sgn L_{j' + i} (f^{(j)}) = \sgn L_{j'+1} (f^{(j)}) \sin \frac{ik-j}r \pi$, it follows that $L_{j'+1} (f^{(j)}) \neq 0$, and so $\sgn L_{j'+1} (f^{(j)}) =  \break \sgn (-1)^{\frac{j'k - j}r + 1}$. Therefore, we have $\sgn L_{j' + i} (f^{(j)}) = \sgn (-1)^{\frac{j'k - j}r + 1} \sin \frac{-ik}r \pi$. Replacing $i$ by $i - j'$, we obtain $\sgn L_{i} (f^{(j)}) = \sgn \sin \frac{j-ik}r \pi$.
\end{proof}

\section{Proof of the Other Results} \label{sec:proofs_of_related_results}

In this section, we prove \Cref{thm:NOW21,thm:subcfrac_converge,prop:unstable}.

\subsection{Proof of \cref{thm:subcfrac_converge}} \label{sec:proof_of_cfrac}

As we pointed out in \cref{sec:cfrac}, the first half and Parts \eqref{item:cfrac_O_lox}, \eqref{item:cfrac_Omega_lox} and \eqref{item:cfrac_O_ell} of the second half of \cref{thm:subcfrac_converge} follow from \cref{thm:main}. We will prove Part \eqref{item:cfrac_Omega_ell} here. 

\begin{proof}[Proof of \cref{thm:subcfrac_converge}~\eqref{item:cfrac_Omega_ell}]
By \cref{prop:cfrac-holonomic}, it suffices to show that $\{ A(n) / B(n) \}_{n \equiv i \pmod{\tau}}$ diverges in $\hat{\R}$ for any $\tau$ and $i$, where $A$ and $B$ are $(P, Q)$-holonomic sequences with initial values $(1, 0)$, $(0, 1)$, respectively. 
Define $P_{\tau}$ and $Q_{\tau}$ as in \eqref{eq:Ptau_and_Qtau}. Then, the subsequences $A(\tau n + i)$ and $B(\tau n + i)$ are $(P_{\tau}(\tau x+i), Q_{\tau}(\tau x+i))$-holonomic sequences. 
From \cref{cor:hyperbolic subsequence}~\eqref{item:Q-Omega_elliptic->Q-Omega_elliptic},
$(P_{\tau}(\tau x+i), Q_{\tau}(\tau x+i))$ is of $\Q$-$\Omega$ elliptic type. 
The divergence of $\{ A(n) / B(n) \}_{n \equiv i \pmod{\tau}}$ follows from \cref{thm:cfrac_cconverge} and \cref{prop:cfrac-holonomic}. 
\end{proof}

\subsection{Proof of \cref{thm:NOW21,prop:unstable}}\label{sec:recovery_of_NOW21}

By the assumption of \cref{thm:NOW21}, we have $\deg \frac{Q(x)}{P(x)P(x-1)} \leq -1$, so $(P, Q)$ is of $\infty$-$\Omega$ loxodromic type or hyperbolic type. 
Then, by \Cref{thm:main}, $(P, Q)$-holonomic sequences $g \in \R^{\N}$ with unstable ultimate signs form a one-dimensional linear subspace in the linear space of all $(P, Q)$-holonomic sequences. 
Therefore, $g(n+1)$ and $g(n)$ must satisfy a linear relation as shown below. To keep the statement simple, let $R(x) := \frac{Q(x)}{P(x)P(x-1)}$ and consider the $(1, R)$-holonomic sequence $f(n) = \frac{g(n)}{P(n-1) \dotsm P(-1)}$ with an unstable ultimate sign instead of $g$.

\begin{lemma}\label{lem:NOW21}
Let $R \in \R(x)$ have no zeros or poles in $\N$ and satisfy $\deg R \leq -1$. Then, for all sufficiently large $n \in \N$, there exists $h(n) \in [ 1-R(n+1) - 3R(n+1)^2, 1-R(n+1) + 3R(n+1)^2 ]$ such that any $(1, R)$-holonomic sequence $f$ whose ultimate sign is unstable satisfies the relation
\begin{equation}\label{eq:f=Rhf}
f(n+1) = - R(n)h(n) f(n). 
\end{equation}
\end{lemma}

The relation~\eqref{eq:f=Rhf} corresponds to the equation~(6) in \cite{NOW21}.
Instead of using \cite[Lemma 14]{NOW21}, whose proof contains a gap, we use \cref{thm:main,lem:generalized_canonical_numerator/denominator} to prove this lemma.

\begin{proof}
Let $A^{(0)}$, $A^{(1)}$, $\dots$ and $B^{(0)}$, $B^{(1)}$, $\dots$ be the generalized canonical numerators and denominators of $(1,R)$. Let $f$ be a $(1, R)$-holonomic sequence whose ultimate sign is unstable. 
Dividing \cref{eq:f=Bf+Af} (with its $Q$ replaced by $R$) by $B^{(i)}(n)$ and using $A^{(i)}(x) = R(x) B^{(i-1)}(x+1)$ in \cref{lem:generalized_canonical_numerator/denominator}, we have $\frac{f(n+i)}{B^{(i)}(n)} = f(n+1) + R(n) \frac{B^{(i-1)}(n+1)}{B^{(i)}(n)} f(n)$. 
Hence showing the existence and estimate of
\begin{equation}\label{eq:h}
h(n) := \lim_{\tau \to \infty} \frac{B^{(\tau-1)}(n+1)}{B^{(\tau)}(n)}
\end{equation}
and $\lim\limits_{\tau \to \infty} \frac{f(n+\tau)}{B^{(\tau)}(n)} = 0$ completes this proof. Take $N \in \N$ such that $|R(n)|$ is monotonically decreasing and less than $\frac19$ for all $n \geq N$.

First, we show that $\frac{B^{(i)}(n+1)}{B^{(i+1)}(n)}$ is contained in the closed interval $[ 1-R(n+1) - 3R(n+1)^2, 1-R(n+1) + 3R(n+1)^2 ]$ with center $1 - R(n+1)$ and radius $3R(n+1)^2$ for all $i \geq 2$ and $n \geq N$, by induction on $i$. We use the inequality
\begin{equation} \label{eq:inequality_for_small_r}
1 - r - 3r^2 \leq \left( 1+r+\frac43 r^2 \right)^{-1} \leq (1+r)^{-1} \leq \left( 1+r-\frac43 r^2 \right)^{-1} \leq 1 - r + 3r^2
\end{equation}
for any $r \in [-\frac19, \frac19]$. 
If $i=2$, then $\frac{B^{(2)}(n+1)}{B^{(3)}(n)} = (1 + R(n+1))^{-1}$. Comparing the very middle of \eqref{eq:inequality_for_small_r} to its very left- and right-hand sides (with $r=R(n+1)$), we get the claim. 
Let us prove the claim for $i+1$, assuming that the claim holds for $i$. 
Replace $(P, Q)$ of \eqref{eq:B} by $(1, R)$, and divide it by $B^{(i+1)}(n+1)$, then we have
\begin{equation} \label{eq:H/H}
\frac{B^{(i+2)}(n)}{B^{(i+1)}(n+1)} = 1 + R(n+1) \frac{B^{(i)}(n+2)}{B^{(i+1)}(n+1)}. 
\end{equation}
$\frac{B^{(i)}(n+2)}{B^{(i+1)}(n+1)}$ in the right-hand side is contained in the closed interval with center $1$ and radius $\frac43 |R(n+1)|$ since $|R(n+2)| \leq |R(n+1)| < \frac19$. 
So the both sides of \eqref{eq:H/H} are in the closed interval with center $1 + R(n+1)$, radius $\frac43 R(n+1)^2$. By the very left ``$\leq$'' and the very right ``$\leq$'' of \eqref{eq:inequality_for_small_r} where $r = R(n+1)$, it follows that $\frac{B^{(i+1)}(n+1)}{B^{(i+2)}(n)}$ is in the closed interval with center $1 - R(n+1)$ and radius $3R(n+1)^2$.

As shown above, $h(n)$ is in the closed interval with center $1-R(n+1)$ and radius $3R(n+1)^2$, if $h(n)$ exists. Next, we prove the existence of $h(n)$. Since $\frac{B^{(i)}(n+1)}{B^{(i+1)}(n)} \in [1-R(n+1)-3R(n+1)^2, 1-R(n+1) + 3R(n+1)^2] \subseteq [\frac12, 2]$ where $n \geq N$ and $i \geq 2$, the existence of $h(n)$ is equivalent to the convergence of the inverse $\frac{B^{(i+1)}(n)}{B^{(i)}(n+1)}$. By \eqref{eq:H/H}, we have 
\[
\begin{aligned}
&\left| \frac{B^{(i+2)}(n)}{B^{(i+1)}(n+1)} - \frac{B^{(i+1)}(n)}{B^{(i)}(n+1)} \right|
\\
&=
|R(n+1)| \left| \frac{B^{(i)}(n+2)}{B^{(i+1)}(n+1)} - \frac{B^{(i-1)}(n+2)}{B^{(i)}(n+1)} \right| 
\\
&=
|R(n+1)| \frac{B^{(i)}(n+2)}{B^{(i+1)}(n+1)} \frac{B^{(i-1)}(n+2)}{B^{(i)}(n+1)} \left| \frac{B^{(i+1)}(n+1)}{B^{(i)}(n+2)} - \frac{B^{(i)}(n+1)}{B^{(i-1)}(n+2)}  \right| 
\\
&\leq
\frac49 \left| \frac{B^{(i+1)}(n+1)}{B^{(i)}(n+2)} - \frac{B^{(i)}(n+1)}{B^{(i-1)}(n+2)} \right|
\\
&\leq \dots \leq
\left( \frac49 \right)^{i-2} \left| \frac{B^{(4)}(n+i-1)}{B^{(3)}(n+i)} - \frac{B^{(3)}(n+i-1)}{B^{(2)}(n+i)} \right| = O\left( \left( \frac49 \right)^{i} \right). 
\end{aligned}
\]
This shows that $\left\{ \frac{B^{(i+1)}(n)}{B^{(i)}(n+1)} \right\}_{i \in \N}$ is a Cauchy sequence and converges.  

Finally we prove $\lim\limits_{i \to \infty} \frac{f(n+i)}{B^{(i)}(n)} = 0$. Recall $\frac{B^{(i)}(n+1)}{B^{(i+1)}(n)} \in [ \frac12 , 2]$ for $n \geq N$ and $i \geq 2$. Then $\frac1{B^{(i)}(n)} = \frac{B^{(i-1)}(n+1)}{B^{(i)}(n)} \frac{B^{(i-2)}(n+2)}{B^{(i-1)}(n+1)} \dotsm \frac{B^{(2)}(n+i-2)}{B^{(3)}(n+i-3)} = O ( 2^i ) \ (i \to \infty)$. Now it remains to show $f(n+i) = O\left( \left( \frac25 \right)^i \right)$, i.e., $f(n) = O\left( \left( \frac25 \right)^n \right) \ (n \to \infty)$. Let $f \neq 0$, since it is obvious if $f=0$. $(1, R)$ is of $\infty$-$\Omega$ loxodromic type or hyperbolic type by the assumption $\deg R \leq -1$. 

Let us first assume that $(1, R)$ is of $\infty$-$\Omega$ loxodromic type. 
It follows from \cref{thm:main}~\eqref{item:loxodromic_far_from_infty} that $f$ has the ultimate sign $(+, -)$ or $(-, +)$. For all $n \geq N$ at which $f$ has the ultimate sign, $R(n)f(n)$ and $f(n+2)$ have the same sign and $f(n+1)$ has the different sign, so it follows from $f(n+2)=f(n+1)+R(n)f(n)$ that $|f(n+2)| < |R(n)f(n)| \leq \frac19 |f(n)|$. Hence $f(n) = O\left( \left( \frac13 \right)^n \right) = O\left( \left( \frac25 \right)^n \right)$. 

Let us second assume that $(1, R)$ is of hyperbolic type. Once $\frac{f(N'+1)}{f(N')} > \frac25$ holds for some $N' \geq N$, then $\frac{f(N'+2)}{f(N'+1)} = 1 + R(N')\frac{f(N')}{f(N'+1)} > \frac{13}{18} > \frac25$, so $\frac{f(n+1)}{f(n)} > \frac25$ holds for all $n \geq N'$. Such $f$ has a stable ultimate sign $(\sgn f(N'))$, which contradicts the assumption of this lemma. Hence $\frac{f(n+1)}{f(n)} \leq \frac25$ for all $n \geq N$. In addition, if $f$ has an ultimate sign at $n$, then $\frac{f(n+1)}{f(n)} > 0$ because the ultimate sign is $(+)$ or $(-)$ according to \cref{thm:main}~\eqref{item:hyperbolic}. These two inequalities imply $f(n) = O\left( \left( \frac25 \right)^n \right)$. 
\end{proof}

We are now ready to prove \cref{thm:NOW21}.

\begin{proof}[Proof of \cref{thm:NOW21}]
Without loss of generality, we can assume $P(-1) \neq 0$. Let us take a $(P, Q)$-holonomic sequence $g \in \Q^{\N}$ with an unstable ultimate sign, and show $g=0$. By multiplying a positive integer by the initial value of $g$, we assume $g \in \Z^{\N}$. Applying \cref{lem:NOW21} to $R(x) := \frac{Q(x)}{P(x)P(x-1)}$ and $f(n) := \frac{g(n)}{P(n-2) \dotsm P(-1)}$, we obtain
\begin{equation}\label{eq:g=Qh/Pg}
g(n+1) = - \frac{Q(n)h(n)}{P(n)} g(n), 
\end{equation}
where $h(n) = 1-\frac{Q(n+1)}{P(n+1)P(n)} + O(n^{-2})$.

\eqref{item:|q|<p} $\left| \frac{Q(n)h(n)}{P(n)} \right| < 1$ holds for all sufficiently large $n$ since $\lim\limits_{n \to \infty} h(n) = 1$. Therefore, $|g(n+1)| < |g(n)|$ or $g(n) = 0$, which implies $g(n) = 0$ for sufficiently large $n$. Since $Q$ has no zeros in $\N$, we get $g=0$. 

\eqref{item:|q|=p} Let us first show $g(n) / n \to 0$. The absolute value of the coefficient in \eqref{eq:g=Qh/Pg} is estimated as
\[
\frac{|Q(n)|h(n)}{P(n)} = 1 + \frac{|Q(n)| - P(n) - \tfrac{|Q(n)|Q(n+1)}{P(n+1)P(n)}}{P(n)} + O(n^{-2}). 
\]
If $d=1$, then this estimate turns out to be $1 + \frac{sq_1-p_1-s}{p_0}n^{-1} + O(n^{-2})$. 
If $d \geq 2$, then $1 + \frac{sq_1-p_1}{p_0}n^{-1} + O(n^{-2})$. Since $\prod_{k=1}^n \left( 1 + \alpha k^{-1} + O(k^{-2}) \right) = O(n^{\alpha})$ for all $\alpha \in \R$, it follows from \eqref{eq:g=Qh/Pg} that
\[
g(n) = 
\begin{cases}
O\!\left( n^{\frac{sq_1-p_1-s}{p_0}} \right) & \text{if $d = 1$}, \\
O\!\left( n^{\frac{sq_1-p_1}{p_0}} \right) & \text{if $d \geq 2$}.
\end{cases}
\]
By the assumption on $p_0$, $p_1$, $q_1$, we have $g(n) / n \to 0$. 

Since $g(n) / n \to 0$ and $d \geq 1$, it follows that $0 = \lim\limits_{n \to \infty} g(n+2) / n^d = \lim\limits_{n \to \infty} (P(n)g(n+1) + Q(n)g(n)) / n^d = \lim\limits_{n \to \infty} (p_0g(n+1) + q_0g(n))$. Since $g(n), g(n+1) \in \Z$, we have $p_0 g(n+1) + q_0 g(n) = 0$ for all sufficiently large $n$. Then $g(n+1) = -sg(n)$ follows from this and $sq_0 = p_0$. Substituting this into the recurrence \eqref{eq:(PQ)holonomic}, we get $g=0$, by the assumption of $Q(x)-sP(x) \neq 1$. 
\end{proof}

We used our stronger assumption to get the equation 
$\lim\limits_{n \to \infty} (P(n)g(n+1) + Q(n)g(n)) / n^d = \lim\limits_{n \to \infty} (p_0g(n+1) + q_0g(n))$
in the last paragraph of the proof above.
We did not make any other changes to the original proof in \cite[\S~3.3]{NOW21}.

We prove \cref{prop:unstable}. 

\begin{proof}[Proof of \cref{prop:unstable}]
Without loss of generality, we can assume $P(-1) \neq 0$. 
The $(P(x), \lambda P(x) + \lambda^2)$-holonomic sequence $\{ (-\lambda)^n \}_{n \in \N}$ has an unstable ultimate sign if and only if the $(1, R)$-holonomic sequence $f := \left\{ \frac{(-\lambda)^n}{P(n-2) \dotsm P(-1)} \right\}_{n \in \N} $ has an unstable ultimate sign, where $R(x) := \frac{\lambda P(x) + \lambda^2}{P(x) P(x-1)}$.

If $\lambda > 0$, then $(1, R)$ is of $\infty$-$\Omega$ loxodromic type. In this case, $f$ has the ultimate sign $(+, -)$, which is unstable by \cref{thm:main}.

If $\lambda < 0$, then $(1, R)$ is of hyperbolic type. 
By a similar argument in the proof of \cref{lem:NOW21}, $(1, R)$-holonomic sequence $g$ with a stable ultimate sign satisfies $|g(n+1) / g(n)| = \Omega(1)$.
On the other hand, $f(n+1) / f(n) = -\lambda / P(n-1) \to 0$. Thus $f$ has an unstable ultimate sign. 
\end{proof}

\bibliography{bib_hol_20251205}

\end{document}